\documentclass[11pt]{article}

\usepackage{amsmath,amsthm,amssymb}
\usepackage{authblk}
\usepackage{fullpage}
\usepackage{graphicx}
\usepackage{caption}
\usepackage{subcaption}
\usepackage[export]{adjustbox}
\usepackage{enumitem}

\title{Isometric embeddings in trees and their use in the diameter problem}
\author[1,2]{Guillaume Ducoffe}
\affil[1]{\small National Institute for Research and Development in Informatics, Romania}
\affil[2]{\small University of Bucharest, Romania}
\date{}

\newtheorem{lemma}{Lemma}
\newtheorem{theorem}{Theorem}

\newtheorem{problem}{\bf Problem}
\newsavebox{\mybox}

\begin{document}

\maketitle

\begin{abstract}
A well-studied topic in Metric Graph Theory, as well as for more general metric spaces, is the existence of isometric embeddings in some product, distribution or system of trees. We initiate the study of a query-answering problem on a finite system of trees whose definition has been tailored for the fast resolution of distance and facility location problems on such ``tree-like'' spaces. Doing so, we prove that given a discrete space with $n$ points which is either embedded in a system of $k$ trees, or the Cartesian product of $k$ trees, we can compute all eccentricities in ${\cal O}(2^{{\cal O}(k\log{k})}(N+n)^{1+o(1)})$ time, where $N$ is the cumulative total order over all these $k$ trees. This is near optimal under the Strong Exponential-Time Hypothesis, even in the very special case of an $n$-vertex graph embedded in a system of $\omega(\log{n})$ spanning trees. However, given such an embedding in the strong product of $k$ trees, there is a much faster ${\cal O}(N + kn)$-time algorithm for this problem. All our positive results can be turned into approximation algorithms for the graphs and finite spaces with a quasi isometric embedding in trees, if such embedding is given as input, where the approximation factor (resp., the approximation constant) depends on the distortion of the embedding (resp., of its stretch). 

The existence of embeddings in the Cartesian product of finitely many trees has been thoroughly investigated for {\em cube-free median graphs}. We give the first-known quasi linear-time algorithm for computing the diameter within this graph class. It does not require an embedding in a product of trees to be given as part of the input. On our way, being given an $n$-node tree $T$, we propose a data structure with ${\cal O}(n\log{n})$ pre-processing time in order to compute in ${\cal O}(k\log^2{n})$ time the eccentricity of any subset of $k$ nodes. We combine the latter technical contribution, of independent interest, with a recent distance-labeling scheme that was designed for cube-free median graphs. 
\end{abstract}

\section{Introduction}\label{sec:intro}

Bandelt and Chepoi identified {\em median graphs}, {\em Helly graphs} and {\em $\ell_1$-graphs} amongst the most important classes from a metric and geometric points of view~\cite{BaC08}. We refer to their paper~\cite{BaC08} for a thorough treatment of these aspects in Graph Theory, and to~\cite{BoM08,Die12} for standard graph terminology. We here address a common property to all three aforementioned graph classes, that is, the existence of isometric embeddings (resp., scale embeddings) from these graphs in some product of trees. Recall that, if $(X,d_X)$ and $(Y,d_Y)$ are metric spaces, then an embedding is simply an injective function $\varphi : X \to Y$. Its distortion is the least $\alpha \geq 1$ s.t., for all $x,x' \in X$, we have $\alpha^{-1} d_X(x,x') \leq d_Y(\varphi(x),\varphi(x')) \leq \alpha d_X(x,x')$. The stretch is the least $\beta \geq 0$ s.t., for all $x,x' \in X$, we have $|d_X(x,x') - d_Y(\varphi(x),\varphi(x'))| \leq \beta$. An isometric embedding is one s.t. $\alpha = 1$, or equivalently $\beta=0$. We are especially interested in {\em graph metrics}. Specifically, being given a graph $G=(V,E,w)$ with positive edge-weights, recall that the weight of a path is the sum of the weights of its edges. The distance $d_G(u,v)$ between two vertices $u,v \in V$ is the smallest weight of a $uv$-path. -- For unweighted graphs, it is equal to the least number of edges of a $uv$-path. -- In particular, $(V,d_G)$ is a metric space, that we call a {\em graph metric}. If furthermore, $G$ is a tree, then we also call it a {\em tree metric}.  Tree metrics are arguably the simplest graph metrics. It is therefore unsurprising that a lot of works in Theoretical Computer Science has been devoted to the embedding of graph metrics in tree metrics with low multiplicative distortion, or even low additive stretch. Throughout this work, we explore some algorithmic applications of these ``tree embeddings'' to distance and facility location problems. We focus on {\em unweighted graphs} but we often mention, at apropriate places throughout the paper, when we can extend our results to weighted graphs, or even more general metric spaces.

\medskip
{\bf Related work and Definitions.}
Before we detail our contributions in the paper, let us summarize the literature about embeddings in {\em one} tree.
The graphs that can be {\em isometrically} embedded in a {\em weighted} tree are exactly the block graphs~\cite{BaM86,How79}. See also~\cite{Ban90,BPS90} for an efficient recognition of tree metrics. More generally, all the metric spaces that embed in a tree with constant distortion have a bounded Gromov hyperbolicity: a polynomial-time computable parameter from geometric group theory that is inspired by the four-point characterization of tree metrics~\cite{Gro87}. Conversely, the Gromov hyperbolicity of a finite metric space is a logarithmic approximation of the best possible distortion of an embedding of this space in a tree~\cite{CDEHV08,Gro87}. For unweighted graphs, there also exist constant-factor approximation algorithms for computing this minimum distortion~\cite{CDNR+12}. However, solving the problem exactly is NP-hard on general metrics~\cite{ABFP+98}. An interesting special case is the embedding of a graph in a spanning tree, that has been studied on its own for various graph classes~\cite{BDLU+07,BDL0L4,CaC95,CaK94,DrK14,FeK01,FGV11,LeL99,MVR96,MRRA03}. This variant of the problem is NP-hard even on chordal graphs~\cite{BDL0L4}, and several inapproximability results have been derived~\cite{EmP09}. To our best knowledge, the best known approximation factor for this problem is logarithmic~\cite{DrK14}.

In this paper, we consider graphs and metric spaces that can be (quasi-)isometrically embedded in a tree {\em or more than one tree}. There are several ways to define such an embedding. For instance, an embedding of a space $(X,d)$ in a system of trees $T_1,T_2,\ldots,T_k$ is defined as $k$ projections $\varphi_i : X \to V(T_i), \ 1 \leq i \leq k$. Then, the distortion of this embedding (resp. its stretch) is defined as the least $\alpha$ s.t., $\forall x,y \in X, \ \alpha^{-1} d(x,y) \leq \min_i d_{T_i}(\varphi_i(x),\varphi_i(y)) \leq \alpha d(x,y)$ (resp., as the least $\beta$ s.t. $\forall x,y \in X, \ |d(x,y) - \min_i d_{T_i}(\varphi_i(x),\varphi_i(y))| \leq \beta$). Every $n$-vertex graph can be isometrically embedded in a system of $n$-shortest-path trees. Even better, there exist quasi-isometric embeddings of many important classes, such as planar and bounded-treewidth graphs, to a sublinear number of spanning trees. Embeddings of graphs in systems of spanning trees were considered in~\cite{DrA14,DCKX12,DXY09,DrY10,DYC06,DYL06,DYX08} for various classes. We further note a distant relation between embedding in a system of trees and embedding in {\em random trees}, as defined in~\cite{Bar96}. Specifically, by a classic application of Markov's inequality, if a metric space with $n$ points can be $\alpha$-probabilistically approximated by tree metrics, then it embeds in a system of ${\cal O}(\log{n}/\varepsilon)$ trees with distortion $\alpha \cdot (1+\varepsilon)$, for any $\varepsilon > 0$. A more broadly studied concept, that we also consider in this paper, is the embedding of metric spaces in some product of trees. We refer to~\cite{HIK11} for the standard definition of several graph products. For instance, the {\em median graphs}, by far the most-studied class in Metric Graph Theory, are exactly the retracts of hypercubes, and so, can be isometrically embedded in a Cartesian product of paths~\cite{Ban84,Isb80,Van83}. Recall that median graphs are exactly the $1$-skeletons of CAT(0) cube complexes~\cite{Gro87}. Furthermore, every {\em cube-free} median graph with maximum degree $\Delta$ can be isometrically embedded in the Cartesian product of $\Delta^{{\cal O}(1)}$ trees~\cite{ChH13}. The graphs that are isometrically embeddable in the Cartesian product of two trees, {\it a.k.a.}, partial double trees, are a subclass of median graphs and they can be recognized in polynomial time~\cite{BCE15}. However, in general, minimizing the number of tree factors necessary for an embedding in a Cartesian product of trees is NP-hard~\cite{BaV89}. In fact, it is already NP-complete to decide whether a median graph embeds in the Cartesian product of three trees. Finally, the {\em Helly graphs}, that are the discrete analogue of hyperconvex metric spaces, are exactly the retracts of strong products of paths~\cite{NoR83}.

\medskip
{\bf Problems considered.} For any point $x$ in some metric space $(X,d)$, let $e(x) := \max_{y \in X} d(x,y)$ be its eccentricity. The diameter and the radius of a metric space are the maximum and minimum eccentricities of a point in that space. We denote the latter by $diam(X,d)$ and $rad(X,d)$ respectively (or, if $G$ is a graph, by $diam(G)$ and $rad(G)$, respectively). Computing all the eccentricities in a finite metric space is a classic facility location problem. Being given the distance-matrix for the $n$ points, it can be done in ${\cal O}(n^2)$ time. In particular, we can compute all the eccentricities in an $n$-vertex $m$-edge graph in total ${\cal O}(nm)$ time. Conversely, assuming the so-called Strong-Exponential-Time Hypothesis (SETH), if an algorithm computes all the eccentricities in a graph, or even just the diameter, in ${\cal O}(n^am^b)$ time, then we must have $a+b \geq 2$~\cite{RoV13}. The problem of finding {\em truly subquadratic} algorithms for the diameter problem on some special graph classes, in the size $n+m$ of the input, has been addressed in many papers, due to its practical and theoretical importance~\cite{BCT17,BCD98,BHM20,CDV02,CDP19,Dam16,FaP80,Ola90}. This question has been solved in the affirmative for Helly graphs~\cite{DuD20+}, but it remains wide open for median graphs~\cite{BCCV20,BCCVR19}. Starting with the pioneering work of Cabello and Knauer~\cite{CaK09}, several recent results have been obtained through the use of tools and parameters from Computational Geometry~\cite{AVW16,BHM20,Cab18,Duc19,DHV19,DHV20,GKHM+18}. In this respect, the celebrated result of Cabello for planar graphs~\cite{Cab18} (subsequently improved in~\cite{GKHM+18}) is especially interesting, since it has also led to exact distance oracles for this graph class with better trade-offs~\cite{CDW17,GMWW18}. In~\cite{Duc20+}, we asked whether we could derive some general conditions under which a distance oracle can be used in order to compute in truly subquadratic time all the eccentricities in a graph, and we gave a partial answer for hub labelings. Since the class of trees admits distance-labeling schemes with logarithmic labels~\cite{GPPR0404}, we may also regard embeddings in trees as another important family of distance oracles.  

\medskip
{\bf Our contributions.} We prove that being given an embedding of a finite metric space in a ``small'' number of trees -- sublogarithmic --, we can break the quadratic barrier for diameter computation. Specifically, we can compute all eccentricities in quasi linear time. Our approach in this paper is inspired by our prior work~\cite{Duc20+} where we obtained similar results for hub labelings. It turns out that {\em all} algorithmic problems that we here study can be reduced to the design of an efficient data structure for some abstract problem over a system of trees. In what follows, $\odot$ denotes a binary associative operation over the nonnegative real numbers ({\it e.g.}, the addition, minimum or maximum of two numbers).

\begin{center}
	\fbox{
		\begin{minipage}{.95\linewidth}
			\begin{problem}[\textsc{$\odot$-Eccentricities}]\
				\label{prob:ecc} 
					\begin{description}
					\item[Global Input:] A system $(T_i)_{1 \leq i \leq k}$ of trees, and a subset $S \subseteq \prod_{i = 1}^k V(T_i)$.
					\item[Query Input:] $v = (v_1,v_2,\ldots,v_k) \in \prod_{i = 1}^k V(T_i)$
					\item[Query Output:] $e_\odot(v,S) := \max\{ d_{T_1}(s_1,v_1) \odot d_{T_2}(s_2,v_2) \odot \ldots  \odot d_{T_k}(s_k,v_k) \mid \\ (s_1,s_2,\ldots,s_k) \in S \}.$
				\end{description}
			\end{problem}     
		\end{minipage}
	}
\end{center}

For instance, let $\varphi$ be an isometric embedding of $(X,d)$ to the Cartesian product of the trees $T_1, T_2, \ldots, T_k$. We set $S := \varphi(X)$, and then, for every $x \in X$ we obtain: $e_{+}(\varphi(x)) = \max_{y \in X} \sum_{i=1}^k d_{T_i}(\varphi_i(x),\varphi_i(y)) = \max_{y \in X}d_{\square_{i=1}^k T_i}(\varphi(x),\varphi(y))$, where $\square$ stands for the Cartesian product and, for every $1 \leq i \leq k$, $\varphi_i$ denotes the projection of $\varphi$ to $V(T_i)$. In particular, $e_{+}(x)$ is equal to the eccentricity of $x$. Similarly, we can reduce the eccentricity problem in systems and strong products of trees to $\min$- and \textsc{$\max$-Eccentricities}, respectively. See Sec.~\ref{sec:red}. We also stress that being given a {\em quasi} isometric embedding, our above approach leads to approximation algorithms for computing all the eccentricities.

\smallskip
In Sec.~\ref{sec:algo}, we report on data structures for \textsc{$\odot$-Eccentricities} where $\odot \in \{\min,\max,+\}$. The pre-processing and query times are {\em quasilinear} and {\em sublinear}, respectively, for every fixed $k$. Our results in this section imply quasilinear-time algorithms for computing all the eccentricities in a finite space $(X,d)$, being given an isometric embedding in a system, strong product or Cartesian product of $k$ trees. Although it is NP-hard to compute such embeddings even for partial cubes~\cite{BaV89}, we stress that there are linear-time algorithms in order to embed {\em hexagonal systems}~\cite{Che96}, {\em squaregraphs}~\cite{BCE10} and {\em partial double trees}~\cite{BCE15} in the Cartesian product of constantly many trees (at most five). -- {\it En passant}, we remark that Chepoi was the first to exploit such embeddings for computing the diameter, but in combination with a ``total monotonic'' property for the distance-matrix of the graphs considered~\cite{Che96}. -- Another interesting add-up to our framework in this paper are the subclasses of median graphs with bounded {\em lattice dimension}, {\it a.k.a.} those median graphs embeddable in the Cartesian product of $k$ paths, for some constant $k$. Indeed, an embedding in the product of a least number of paths can be computed in ${\cal O}(n\tau + \tau^{2.5})$ time, with $\tau$ the {\em isometric dimension} -- smallest dimension of a hypercube in which the graph isometrically embeds --~\cite{Che12,Epp05}. Thus, for all the median graphs with bounded lattice dimension, we can compute all the eccentricities in truly subquadratic time, provided their isometric dimension is in ${\cal O}(n^{4/5 - \epsilon})$ for some $\epsilon > 0$.  

We complete our positive results in Sec.~\ref{sec:algo} with simple {\em conditional lower bounds} showing their near optimality. See Sec.~\ref{sec:hardness}. Finally, in Sec.~\ref{sec:kEccTree}, we address the problem of computing the eccentricity of $k$-subsets of nodes in a tree, with an application to the diameter problem for {\em cube-free median graphs}. Given a tree $T$, the eccentricity of a node-subset $U$ is defined as $e_T(U) = \max_{v \in V(T)}\min_{u \in U} d_T(v,u)$. We note that an efficient data structure for computing the eccentricities of $k$-subsets of nodes may be of potential interest for $k$-facility location problems~\cite{KhS00,Tin84}. The latter can be reduced to the design of a data structure for \textsc{$\min$-Eccentricities} (roughly, we just need to consider a system of $k$ copies of $T$). However, for this special case, we obtain better trade-offs than in Sec.~\ref{sec:algo}, namely: an ${\cal O}(|V(T)|\log{|V(T)|})$ pre-processing time (independent of $k$), and an ${\cal O}(k\log^2{|V(T)|})$ query time. See Theorem~\ref{thm:main-tree}. These running times are optimal up to polylogarithmic factors.
Then, recall that the {\em dimension} of a median graph $G$ is the largest $d \geq 1$ such that $G$ contains a $d$-cube (hypercube of dimension $d$) as an isometric subgraph. In particular, the median graphs of dimension $1$ are exactly the trees. The median graphs of dimension at most $2$, {\it a.k.a.}, {\em cube-free median graphs}, have already received some attention in the literature~\cite{BCE15,BKS07,ChH13,CLR20,ChM13}.
Chepoi et al. recently proposed a distance-labeling scheme for this class~\cite{CLR20}. A central object in their construction are the so-called pseudo-gated trees. By applying Theorem~\ref{thm:main-tree} to the latter, we show how to derive a quasi linear ${\cal O}(n\log^3{n})$-time algorithm for the diameter problem on cube-free median graphs (Theorem~\ref{thm:cube-free-median}). 

\section{Reductions}\label{sec:red}

We provide reductions from the various problems studied in this paper to the design of an efficient data structure for \textsc{$\odot$-Eccentricities} (Problem~\ref{prob:ecc}). Note that most proofs in this section are very similar to each other.

\begin{lemma}\label{lem:reduction-collection}
Let $(X,d)$ be a metric space with $n$ points, and let $\varphi$ be an embedding of that space in a system of $k$ trees with distortion $\alpha$ (resp., with stretch $\beta$), for some $k \geq 1$. Denote these $k$ trees by $T_1,T_2,\ldots,T_k$ and set $N = \sum_{i=1}^k |V(T_i)|$. If we can solve \textsc{$\min$-Eccentricities} with $C_p(N,|S|,k)$ pre-processing time and $C_q(N,|S|,k)$ query time, then in total ${\cal O}(C_p(N,n,k)+n\cdot(C_q(N,n,k)+k))$ time we can compute an $\alpha^2$-approximation of all eccentricities of $(X,d)$ (resp., an $+2\beta$-approximation).

In particular, if $\varphi$ is an {\em isometric embedding}, then we can compute the diameter and all eccentricities of $(X,d)$ in ${\cal O}(C_p(N,n,k)+n\cdot(C_q(N,n,k)+k))$ time.
\end{lemma}

\begin{proof}
Recall that, for every $x \in X$, we have $\varphi(x) = (\varphi_1(x),\varphi_2(x),\ldots,\varphi_k(x))$, where $\forall 1 \leq i \leq k, \ \varphi_i(x) \in V(T_i)$. Set $S = \{ \varphi(x) \mid x \in X \}$. Note that being given $\varphi$, the set $S$ can be constructed in total ${\cal O}(kn)$ time. Then, w.r.t. the collection $(T_i)_{1 \leq i \leq k}$ and $S$: 
$$\forall x \in X, \ e_{\min}(\varphi(x),S) = \max_{y \in X}\min_{1 \leq i \leq k} d_{T_i}(\varphi_i(x),\varphi_i(y)).$$
Note that, after a pre-processing in $C_p(N,n,k)$ time, we can compute any value $e_{\min}(\varphi(x),S)$ in $C_q(N,n,k)$ time.
In particular, $\alpha^{-1}e(x) \leq e_{\min}(\varphi(x),S) \leq \alpha e(x)$ (resp., $|e(x) - e_{\min}(\varphi(x),S)| \leq \beta$). As a result, $\alpha e_{\min}(\varphi(x),S)$ is an $\alpha^2$-approximation of $e(x)$ (resp., $e_{\min}(\varphi(x),S) + \beta$ is an $+2\beta$ approximation of $e(x)$).
\end{proof}

\begin{lemma}\label{lem:reduction-cartesian}
Let $(X,d)$ be a metric space with $n$ points, and let $\varphi$ be an embedding of that space in a Cartesian product of $k$ trees with distortion $\alpha$ (resp., with stretch $\beta$), for some $k \geq 1$. Denote these $k$ tree factors by $T_1,T_2,\ldots,T_k$ and set $N = \sum_{i=1}^k |V(T_i)|$. If we can solve \textsc{$+$-Eccentricities} with $C_p(N,|S|,k)$ pre-processing time and $C_q(N,|S|,k)$ query time, then in total ${\cal O}(C_p(N,n,k)+n\cdot(C_q(N,n,k)+k))$ time we can compute an $\alpha^2$-approximation of all eccentricities of $(X,d)$ (resp., an $+2\beta$-approximation).

In particular, if $\varphi$ is an {\em isometric embedding}, then we can compute the diameter and all eccentricities of $(X,d)$ in ${\cal O}(C_p(N,n,k)+n\cdot(C_q(N,n,k)+k))$ time.
\end{lemma}

\begin{proof}
As for Lemma~\ref{lem:reduction-collection}, we recall that, for every $x \in X$, we have $\varphi(x) = (\varphi_1(x),\varphi_2(x),\ldots,\varphi_k(x))$, where $\forall 1 \leq i \leq k, \ \varphi_i(x) \in V(T_i)$. Set $S = \{ \varphi(x) \mid x \in X \}$ (constructible in ${\cal O}(kn)$ time). Then, w.r.t. the collection $(T_i)_{1 \leq i \leq k}$ and the set $S$: 
$$\forall x \in X, \ e_{+}(\varphi(x),S) = \max_{y \in X}\sum_{i=1}^k d_{T_i}(\varphi_i(x),\varphi_i(y)) = \max_{y \in X}d_{\square_{i=1}^k T_i}(\varphi(x),\varphi(y)).$$
Note that, after a pre-processing in $C_p(N,n,k)$ time, we can compute any value $e_{+}(\varphi(x),S)$ in $C_q(N,n,k)$ time.
In particular, $\alpha^{-1}e(x) \leq e_{+}(\varphi(x),S) \leq \alpha e(x)$ (resp., $|e(x) - e_{+}(\varphi(x),S)| \leq \beta$). As a result, $\alpha e_{+}(\varphi(x),S)$ is an $\alpha^2$-approximation of $e(x)$ (resp., $e_{+}(\varphi(x),S) + \beta$ is an $+2\beta$ approximation of $e(x)$).
\end{proof}

\begin{lemma}\label{lem:reduction-strong}
Let $(X,d)$ be a metric space with $n$ points, and let $\varphi$ be an embedding of that space in a strong product of $k$ trees with distortion $\alpha$ (resp., with stretch $\beta$), for some $k \geq 1$. Denote these $k$ tree factors by $T_1,T_2,\ldots,T_k$ and set $N = \sum_{i=1}^k |V(T_i)|$. If we can solve \textsc{$\max$-Eccentricities} with $C_p(N,|S|,k)$ pre-processing time and $C_q(N,|S|,k)$ query time, then in total ${\cal O}(C_p(N,n,k)+n\cdot(C_q(N,n,k)+k))$ time we can compute an $\alpha^2$-approximation of all eccentricities of $(X,d)$ (resp., an $+2\beta$-approximation).

In particular, if $\varphi$ is an {\em isometric embedding}, then we can compute the diameter and all eccentricities of $(X,d)$ in ${\cal O}(C_p(N,n,k)+n\cdot(C_q(N,n,k)+k))$ time.
\end{lemma}

\begin{proof}
We set $S = \{ \varphi(x) \mid x \in X \}$, as for Lemmas~\ref{lem:reduction-collection} and~\ref{lem:reduction-cartesian}. Then, w.r.t. the collection $(T_i)_{1 \leq i \leq k}$ and the set $S$: 
$$\forall x \in X, \ e_{\max}(\varphi(x),S) = \max_{y \in X}\max_{1 \leq i \leq k} d_{T_i}(\varphi_i(x),\varphi_i(y)) = \max_{y \in X}d_{\odot_{i=1}^k T_i}(\varphi(x),\varphi(y)),$$
where $\odot$ stands for the strong product operator.
We conclude as for the two previous lemmas.
\end{proof}

Finally, we emphasize the following other special case of \textsc{$\min$-Eccentricities}:

\begin{lemma}\label{lem:reduction-k-ecc}
Let $T$ be a tree with $n$ nodes and let $k \geq 1$ be an integer. 
If we can solve \textsc{$\min$-Eccentricities} with $C_p(N,|S|,k)$ pre-processing time and $C_q(N,|S|,k)$ query time, where $N$ is the sum of the orders of the $k$ trees in the input system, then 
after a pre-processing of $T$ in ${\cal O}(kn + C_p(kn,n,k))$ time, for any $k$-subset of nodes $U$ we can compute its eccentricity in ${\cal O}(k+C_q(kn,n,k))$ time.
\end{lemma}

\begin{proof}
For every $1 \leq i \leq k$, let $T_i$ be a copy of $T$, and let $\varphi_i : V(T) \to V(T_i)$ be an isomorphism. We set $S = \{ (\varphi_1(v),\varphi_2(v),\ldots,\varphi_k(v)) \mid v \in V(T)\}$. In particular, the elements of the set $S$ are, for each node of $T$, the collection of its $k$ copies. Furthermore, for every subset $U$ of $k$ nodes, let us define $\varphi(U) = (\varphi_1(u_1),\varphi_2(u_2),\ldots,\varphi_k(u_k))$, where $u_1,u_2,\ldots,u_k$ are the nodes of this subset arbitrarily ordered. By construction, we have:
$$e_{\min}(\varphi(U),S) = \max_{v \in V(T)}\min_{1 \leq i \leq k} d_{T_i}(\varphi_i(u_i),\varphi_i(v)) = \max_{v \in V(T)}\min_{1 \leq i \leq k}d_T(u_i,v) = e_T(U).$$
\end{proof}

\section{Algorithms}\label{sec:algo}

Our main result in this section is as follows:

\begin{theorem}\label{thm:main}
For every $(T_i)_{1 \leq i \leq k}$ and $S$, let $N := \sum_{i=1}^k |V(T_i)|$. We can solve \textsc{$\min$-Eccentricities} (resp., \textsc{$+$-Eccentricities}) with ${\cal O}(2^{{\cal O}(k\log{k})}(N+|S|)^{1+o(1)})$ pre-processing time and ${\cal O}(2^{{\cal O}(k\log{k})}(N+|S|)^{o(1)})$ query time.
\end{theorem}

We need to introduce two useful tools. First, let $V$ be a set of $k$-dimensional points, and let $f : V \to \mathbb{R}$. A {\em box} is the Cartesian product of $k$ intervals. A range query asks, given a box ${\cal R}$, for a point $\overrightarrow{p} \in V \cap {\cal R}$ s.t. $f(\overrightarrow{p})$ is maximized. Up to some pre-processing in ${\cal O}(|V|\log^{k-1}{|V|})$ time, such query can be answered in ${\cal O}(\log^{k-1}{|V|})$ time~\cite{Wil85}. The corresponding data structure is called a $k$-dimensional range tree. Furthermore, note that for any $\varepsilon > 0$, $\forall x > 0, \ \log^k{x} \leq 2^{{\cal O}(k\log{k})}x^{\varepsilon}$~\cite{BHM20}. 

Second, for an $n$-node tree $T=(V,E)$, a {\em centroid} is a node whose removal leaves subtrees of order at most $n/2$. A classic theorem from Jordan asserts that such node always exists~\cite{Jor69}. Furthermore, we can compute a centroid in ${\cal O}(n)$ time by dynamic programming ({\it e.g.}, see~\cite{Gol71}). A {\em centroid decomposition} of $T$ is a rooted tree $T'$, constructed as follows. If $|V(T)| \leq 1$, then $T' = T$. Otherwise,  let $c$ be a centroid. Let $T_1',T_2',\ldots,T_p'$ be centroid decompositions for the subtrees $T_1,T_2,\ldots,T_p$ of $T \setminus \{c\}$. We obtain $T'$ from $T_1',T_2',\ldots,T_p'$ by adding an edge between $c$ and the respective roots of these rooted subtrees, choosing $c$ as the new root. Note that we can compute a centroid decomposition in ${\cal O}(n)$ time~\cite{DPV19}. We rather use the folklore ${\cal O}(n\log{n})$-time algorithm since, for any node $v$, we also want to store its path $P(v)$, in $T'$, until the root, and the distances $d_T(v,c_i)$, in $T$, for any $c_i \in P(v)$.

\begin{proof}[Proof of Theorem~\ref{thm:main}]
We first give a proof for the case $\odot = +$. During a pre-processing phase, we compute a centroid decomposition for each tree $T_i, \ 1 \leq i \leq k$ separately. Furthermore, let $T_i'$ be the resulting centroid decomposition of $T_i$. For every node $v_i \in V(T_i)$, we compute the path $P(v_i)$ from $v_i$ until the root of $T_i'$. Note that the length of $P(v_i)$ is at most the depth of $T_i'$, and so, it is in ${\cal O}(\log{N})$. We also compute, for every $c_i \in P(v_i)$, the distance $d_{T_i}(v_i,c_i)$ in the original tree $T_i$. By using the folklore algorithm in order to compute the centroid decomposition ({\it e.g.}, see~\cite{GPPR0404}), all the above computations can be done in total ${\cal O}(N\log{N})$ time. Then, we iterate over the elements $s = (s_1,s_2,\ldots,s_k) \in S$. For every $1 \leq i \leq k$, let $P(s_i)$ be the path of $s_i$ to the root into the centroid decomposition $T_i'$ computed for $T_i$. We consider all possible $k$-sequences $c = (c_1,c_2,\ldots,c_k)$ s.t., $\forall 1 \leq i \leq k$, $c_i \in P(s_i)$. There are ${\cal O}(\log^k{N})$ possibilities. W.l.o.g., all nodes have a unique identifier. Throughout the remainder of the proof, we identify the nodes with their identifiers, thus treating them as numbers. For any sequence $c$, we create an $2k$-dimensional point $\overrightarrow{p}_c(s)$, as follows: for every $1 \leq i \leq k$, the $(2i-1)^{th}$ and $2i^{th}$ coordinates are equal to $c_i$ and the unique neighbour of $c_i$ onto the $c_is_i$-path in $T_i'$, respectively (if $c_i = s_i$, then we may set both coordinates equal to $s_i$). The construction of all these ${\cal O}(|S|\log^k{N})$ points takes time ${\cal O}(k|S|\log^k{N})$. We include these points $\overrightarrow{p}_c(s)$, with an associated value $f(\overrightarrow{p}_c(s))$ (to be specified later in the proof) in some $2k$-dimensional range tree. It takes ${\cal O}((|S|\log^k{N})\log^{2k-1}{(|S|\log^k{N})})$ time.

Then, in order to answer a query, let $v=(v_1,v_2,\ldots,v_k)$ be the input. As before, for every $1 \leq i \leq k$, let $P(v_i)$ be the path of $v_i$ to the root into $T_i'$. We iterate over all the $k$-sequences $c = (c_1,c_2,\ldots,c_k)$ s.t., $\forall 1 \leq i \leq k$, $c_i \in P(v_i)$. Let $S_c \subseteq S$ contain every $(s_1,s_2,\ldots,s_k)$ s.t. $\forall 1 \leq i \leq k$, $c_i$ is the least common ancestor of $v_i$ and $s_i$ in $T_i'$. The subsets $S_c$ partition $S$, and so, $e_+(v,S) = \max_{c} e_+(v,S_c)$. We are left explaining how to compute $e_+(v,S_c)$ for any fixed $c$. For that, we observe that $c_i$ is the least common ancestor of $v_i$ and $s_i$ in $T_i'$ if and only if $c_i \in P(s_i) \cap P(v_i)$, and either $c_i \in \{v_i,s_i\}$ or the two neighbours of $c_i$ onto the $c_iv_i$-path and $c_is_i$-path in $T_i'$ are different. In the latter case, let us denote by $u_i$ the neighbour of $c_i$ onto the $c_iv_i$-path in $T_i'$. We can check these above conditions, for every $1 \leq i \leq k$, with the following constraints, over the $2k$-dimensional points $\overrightarrow{p} = (p_1,p_2,\ldots,p_{2k})$ constructed during the pre-processing phase: $\forall 1 \leq i \leq k, \ p_{2i-1} = c_i$ and if $c_i \neq v_i, \ p_{2i} \neq u_i$. Since $p_{2i} \neq u_i$ is equivalent to $p_{2i} \in (-\infty,u_i) \cup (u_i,+\infty)$, each inequality can be replaced by two range constraints over the same coordinate. In particular, since there are $\leq k$ such inequalities, we can transform these above constraints into ${\cal O}(2^k)$ range queries. Therefore, we can output a point $\overrightarrow{p}_c(s)$, for some $s \in S_c$, maximizing $f(\overrightarrow{p}_c(s))$, in ${\cal O}(2^k\log^{2k-1}{(|S|\log^k{N})})$ time. Let $f(\overrightarrow{p}_c(s)) = \sum_{i=1}^k d_{T_i}(s_i,c_i)$. Since we have $s \in S_c$, $d_{T_i}(s_i,v_i) = d_{T_i}(s_i,c_i) + d_{T_i}(c_i,v_i)$~\cite{GPPR0404}. As a result: $e_+(v,S_c) = f(\overrightarrow{p}_c(s)) + \sum_{i=1}^k d_{T_i}(c_i,v_i)$. There are ${\cal O}(\log^k{N})$ possible $c$, and so, the final query time is in ${\cal O}(2^k\log^{2k-1}{(|S|\log^k{N})}\log^k{N})$.

\medskip
For the case when $\odot = \min$, we need points with ${\cal O}(k)$ more coordinates in order to correctly identify some index $i$ s.t. $d_{T_i}(v_i,s_i) = \min_{1 \leq j \leq k} d_{T_j}(v_j,s_j)$. This is a similar trick as the one used in~\cite{AVW16,BHM20,CaK09}. Specifically, we replace each point $\overrightarrow{p}_c(s)$ by $k$ different $3k$-dimensional points $\overrightarrow{q}_{c,i}(s), \ 1 \leq i \leq k$. The point $\overrightarrow{q}_{c,i}(s)$ is obtained from $\overrightarrow{p}_c(s)$ by adding as its last $k$ coordinates the index $i$ and the values $d_{T_j}(s_j,c_j) - d_{T_i}(s_i,c_i), \ j \neq i$. During a pre-processing phase, we add these ${\cal O}(k|S|\log^k{N})$ points to some $3k$-dimensional range tree, with some associated value $f(\overrightarrow{q}_{c,i}(s))$ to be defined later in the proof. It can be done in total ${\cal O}((k|S|\log^k{N})\log^{3k-1}{(k|S|\log^k{N})})$ time.
In order to answer a query for $v=(v_1,v_2,\ldots,v_k)$, we iterate over the subsets $S_c$ as they were earlier defined in the proof. In particular, we are left explaining how to compute $e_{\min}(v,S_c)$ for any fixed $c$. For that, we reuse our prior range queries over the points $\overrightarrow{p}_c(s)$, but we need to complete the latter in order to determine some index $i$ s.t. $d_{T_i}(v_i,s_i) = \min_{1 \leq j \leq k} d_{T_j}(v_j,s_j)$. Specifically, for every $1 \leq i \leq k$, let $S_{c,i} \subseteq S_c$ contain every $(s_1,s_2,\ldots,s_k)$ s.t. $d_{T_i}(v_i,s_i) = \min_{1 \leq j \leq k} d_{T_j}(v_j,s_j)$. The subsets $S_{c,i}$ cover $S_c$, and so, $e_{\min}(v,S_c) = \max_{1 \leq i \leq k} e_{\min}(v,S_{c,i})$. Next, we explain how to compute $e_{\min}(v,S_{c,i})$ for any fixed $i$ (the total running time in order to compute $e_{\min}(v,S_c)$ is the same as for computing $e_{\min}(v,S_{c,i})$, up to an ${\cal O}(k)$ factor). Observe that, for $(s_1,s_2,\ldots,s_k) \in S_c$, we have: 
\begin{align*}
d_{T_i}(v_i,s_i) \leq d_{T_j}(v_j,c_j) &\Longleftrightarrow d_{T_i}(v_i,c_i) + d_{T_i}(c_i,s_i) \leq d_{T_j}(v_j,c_j) + d_{T_j}(c_j,s_j) \\
&\Longleftrightarrow d_{T_i}(s_i,c_i) - d_{T_j}(s_j,c_j) \leq d_{T_j}(v_j,c_j) - d_{T_i}(v_i,c_i).
\end{align*}
Therefore, we can check all the required conditions for $S_{c,i}$ by using the following constraints over the points $\overrightarrow{q} = (q_1,q_2,\ldots,q_{3k})$ -- constructed during the pre-processing phase:
$$\begin{cases}
\forall 1 \leq j \leq k \ q_{2j-1} = c_j, \\
\forall 1 \leq j \leq k \ \text{if} \ c_j \neq v_j, \ q_{2j} \neq u_j, \\
q_{2k+1} = i, \\
\forall 1 \leq j \leq i-1, \ q_{2k+1+j} \leq d_{T_j}(v_j,c_j) - d_{T_i}(v_i,c_i), \\
\forall i+1 \leq j \leq k, \ q_{2k+j} \leq d_{T_j}(v_j,c_j) - d_{T_i}(v_i,c_i).
\end{cases}$$
As in the first part of the proof, the above can be transformed into ${\cal O}(2^k)$ range queries.
Then, we can output a point $\overrightarrow{q}_{c,i}(s)$, for some $s \in S_{c,i}$, maximizing $f(\overrightarrow{q}_{c,i}(s))$, in ${\cal O}(2^k\log^{3k-1}{(k|S|\log^k{N})})$ time. In particular, let $f(\overrightarrow{q}_{c,i}(s)) = d_{T_i}(c_i,s_i)$. We have: 
$$e_{\min}(v,S_{c,i}) = f(\overrightarrow{q}_{c,i}(s)) + d_{T_i}(c_i,v_i).$$

\medskip
\noindent
{\bf Complexity analysis.}
Observe that the running time for $\odot = \min$ is higher than for $\odot = +$ (although, as it can be deduced from our analysis below, they are asymptotically comparable). Therefore, let us further analyse the running time of our algorithm for the case when $\odot = \min$. 

\smallskip
The pre-processing time is in ${\cal O}(N\log{N} + (k|S|\log^k{N})\cdot\log^{3k-1}{(k|S|\log^k{N})})$. Recall that, for any $\varepsilon > 0$, we have $\forall x > 0, \ \log^k{x} \leq 2^{{\cal O}(k\log{k})}x^{\varepsilon}$~\cite{BHM20}. In particular, for any $\varepsilon >0$: 
\begin{align*}
N\log{N} + (k|S|\log^k{N})\cdot\log^{3k-1}{(k|S|\log^k{N})} &\leq N^{1+\varepsilon} + (k|S|\log^k{N}) \cdot (2^{{\cal O}(k\log{k})}(k|S|\log^k{N})^\varepsilon) \\
&= N^{1+\varepsilon} + 2^{{\cal O}(k\log{k})} (k|S|\log^k{N})^{1+\varepsilon} \\
&\leq N^{1+\varepsilon} + 2^{{\cal O}(k\log{k})} \cdot (k2^{{\cal O}(k\log{k})} |S|N^{\varepsilon})^{1+\varepsilon} \\
&= N^{1+\varepsilon} +  2^{{\cal O}(k\log{k})} |S|^{1+\varepsilon} N^{{\cal O}(\varepsilon)} \\
&\leq N^{1+\varepsilon} +  2^{{\cal O}(k\log{k})} |S| \cdot (|S| + N)^{{\cal O}(\varepsilon)} \\
&\leq 2^{{\cal O}(k\log{k})} (N+|S|)^{1+{\cal O}(\varepsilon)}.
\end{align*}
As a result, the pre-processing time is in ${\cal O}(2^{{\cal O}(k\log{k})}(N+|S|)^{1+o(1)})$.

\smallskip
The query time is in ${\cal O}(k\cdot (\log^k{N}) \cdot 2^k\log^{3k-1}{(k|S|\log^k{N})})  = 2^{{\cal O}(k\log{k})}(N + |S|)^{o(1)}$.
\end{proof}

In contrast to Theorem~\ref{thm:main}, the distance between two vertices in the strong product of $k$ trees equals the maximum distance between their $k$ respective projections. Hence, we can process each tree of the system separately, and we obtain:

\begin{lemma}\label{lem:strong-prod}
For every $(T_i)_{1 \leq i \leq k}$ and $S$, let $N := \sum_{i=1}^k |V(T_i)|$. We can solve \textsc{$\max$-Eccentricities} with ${\cal O}(N + k|S|)$ pre-processing time and ${\cal O}(k)$ query time.
\end{lemma}

\begin{proof}
For every $1 \leq i \leq k$, let $\varphi_i : S \to V(T_i)$ be the projection of $S$ to $T_i$. We stress that the projections $\varphi_i(S), \ 1 \leq i \leq k$, can be computed in total ${\cal O}(k|S|)$ time.
Then, we iteratively remove from $T_i$ the leaves that are not in $\varphi_i(S)$. Let $T_1',T_2',\ldots,T_k'$ be the $k$ subtrees resulting from this above pre-processing. Note that, for every $1 \leq i \leq k$, $T_i \setminus T_i'$ is a forest whose each subtree can be rooted at some node adjacent to a leaf of $T_i'$, and so, adjacent to a node of $\varphi_i(S)$. For every node $v_i \in V(T_i) \setminus V(T_i')$, let $\phi(v_i)$ be the unique leaf of $T_i'$ s.t. the subtree of $T_i \setminus T_i'$ that contains $v_i$ also contains a neighbour of $\phi(v_i)$. We compute, and store, the distance $d_{T_i}(v_i,\phi(v_i))$. Since, for doing so, we only need to perform breadth-first searches on disjoint subtrees, the total running time of this step is in ${\cal O}(N)$. Finally, we compute, for $1 \leq i \leq k$, all eccentricities in $T_i'$. Again, this can be done in total ${\cal O}(N)$ time ({\it e.g.}, see~\cite{CDHV+19,Duc19b}). This concludes the pre-processing phase.

In order to answer a query, let us consider some input $v = (v_1,v_2,\ldots,v_k)$. Our key insight here is that we have:
$$e_{\max}(v,S) = \max\{ \max_{1\leq i \leq k} d_{T_i}(v_i,s_i) \mid (s_1,s_2,\ldots,s_k) \in S \} = \max_{1 \leq i \leq k}\max\{ d_{T_i}(v_i,s_i) \mid s_i \in \varphi_i(S) \}.$$
Therefore, in order to compute $e_{\max}(v,S)$ in ${\cal O}(k)$ time, it suffices to compute $\max\{ d_{T_i}(v_i,s_i) \mid s_i \in \varphi_i(S) \}$ in ${\cal O}(1)$ time for every $1 \leq i \leq k$. If $v_i \in V(T_i')$ then, since $\varphi_i(S) \subseteq V(T_i')$ and furthermore all the leaves of $T_i'$ are in $\varphi_i(S)$, we get $\max\{ d_{T_i}(v_i,s_i) \mid s_i \in \varphi_i(S) \} = e_{T_i'}(v_i)$. Otherwise, $\max\{ d_{T_i}(v_i,s_i) \mid s_i \in \varphi_i(S) \} = d_{T_i}(v_i,\phi(v_i)) + e_{T_i'}(\phi(v_i))$.
\end{proof}

\section{Hardness results}\label{sec:hardness}

We complete the positive results of Theorem~\ref{thm:main} with two conditional lower bounds. Recall that a split graph is a graph whose vertex-set can be bipartitioned into a clique $K$ and a stable set $I$. The Strong Exponential-Time Hypothesis (SETH) says that for any $\varepsilon > 0$, there exists a $k$ such that $k$-{\sc SAT} on $n$ variables cannot be solved in ${\cal O}((2-\varepsilon)^n)$ time~\cite{ImP01}. Both theorems  in this section follow from a ``SETH-hardness'' result in order to compute the diameter of {\em split graphs} with a logarithmic clique-number~\cite{BCH16}, and from the observation that every split graph with a maximal clique $K$ embeds in any system of $|K|$ shortest-path trees rooted at the vertices of $K$. Specifically:

\begin{lemma}[\cite{BCH16}]\label{lem:ov}
For any $\varepsilon > 0$, there exists a $c(\varepsilon)$ s.t., under SETH, we cannot compute the diameter in ${\cal O}(n^{2-\varepsilon})$ time on the split graphs of order $n$ and clique-number at most $c(\varepsilon) \log{n}$.
\end{lemma}

We deduce from this above Lemma~\ref{lem:ov} the following two ``SETH-hardness'' results:

\begin{theorem}\label{thm:collection-hardness}
For any $\varepsilon > 0$, there exists a $c(\varepsilon)$ s.t., under SETH, we cannot compute the diameter of $n$-vertex graphs in ${\cal O}(n^{2-\varepsilon})$ time, even if we are given as input an isometric embedding of the graph in a system of at most $c(\varepsilon) \log{n}$ spanning trees. 
In particular, under SETH, there is no data structure for \textsc{$\min$-Eccentricities} with ${\cal O}(2^{o(k)}(N+|S|)^{1-o(1)})$ pre-processing time and ${\cal O}(2^{o(k)}(N+|S|)^{o(1)})$ query time, where $N := \sum_{i=1}^k |V(T_i)|$. 
\end{theorem}

\begin{proof}
Let $G=(K \cup I,E)$ be a split graph, where $K$ and $I$ are a clique and a stable set, respectively. Note that such a bipartition of $V(G)$ can be computed in linear time~\cite{Gol04}. Let $(T_u)_{u \in K}$ be arbitrary shortest-path trees rooted at every vertex of the clique $K$. The trivial embedding of $G$ into the collection $(T_u)_{u \in K}$, associating to each vertex $v$ its $|K|$ copies in these shortest-path trees, is isometric. If furthermore, $|K| \leq c \log{n}$ for some $c > 0$, then this above collection can be computed in total ${\cal O}(n \log^2{n})$ time (quasi linear). The result now follows from Lemma~\ref{lem:ov}.
\end{proof}

We believe our Theorem~\ref{thm:collection-hardness} to be important since the embeddings of graphs in systems of tree spanners are well-studied in the literature~\cite{DrA14,DCKX12,DXY09,DrY10,DYC06,DYL06,DYX08}.

\begin{theorem}\label{thm:cartesian-hardness}
For any $\varepsilon > 0$, there exists a $c(\varepsilon)$ s.t., under SETH, we cannot compute the diameter of $n$-point metric spaces in ${\cal O}(n^{2-\varepsilon})$ time, even if we are given as input an isometric embedding of the space in a Cartesian product of at most $c(\varepsilon) \log{n}$ tree factors. 
In particular, under SETH, there is no data structure for \textsc{$+$-Eccentricities} with ${\cal O}(2^{o(k)}(N+|S|)^{1-o(1)})$ pre-processing time and ${\cal O}(2^{o(k)}(N+|S|)^{o(1)})$ query time, where $N := \sum_{i=1}^k |V(T_i)|$. 
\end{theorem}

\begin{proof}
As for Theorem~\ref{thm:collection-hardness}, let  $G=(K \cup I,E)$ be a split graph, where $K$ and $I$ are a clique and a stable set, respectively. Again, note that such a bipartition of $V(G)$ can be computed in linear time~\cite{Gol04}. For every $u \in K$, we construct a tree $T'_u$, and an embedding $\varphi_u$ of $G$ into the latter, as follows. We start from a single-node tree, to which we map vertex $u$. Then, for every $v \in N_G(u)$, we add a leaf into the tree, adjacent to the image of $u$, to which we map the vertex $v$. We add another node $u^*$ in $T'_u$, that is also adjacent to the image of $u$ (note that $u^*$ is {\em not} the image of a vertex of $G$). Finally, for every vertex $v \in V(G) \setminus N_G[u]$, we add a leaf node, to which we map vertex $v$, that we connect to $u^*$ by a path of length two. Let $\varphi : V(G) \to V(\square_{u \in K} T'_u)$ be s.t., for every $v \in V(G)$, $\varphi(v) = (\varphi_u(v))_{u \in K}$. The metric space considered is $(\varphi(V(G)),d)$, where $d$ is the sub-metric induced by $d_{\square_{u \in K} T'_u}$ (distances in the Cartesian product). Indeed, let $v,v' \in V(G)$. For every $u \in K$, by construction we have $diam(T_u') = 4$, and so, $d_{T_u'}(\varphi_u(v),\varphi_u(v')) \leq 4$. Specifically, if $u \in \{v,v'\}$ then $d_{T_u'}(\varphi_u(v),\varphi_u(v')) \leq 3$; if $v,v' \in N_G(u)$ then $d_{T_u'}(\varphi_u(v),\varphi_u(v'))  = 2$; otherwise, $dist_{T_u'}(\varphi_u(v),\varphi_u(v')) = 4$. Altogether combined, if $d_G(v,v') = 3$ (in particular, $v,v' \in I$) then we get $d(\varphi(v),\varphi(v')) = 4|K|$, otherwise we get $d(\varphi(v),\varphi(v')) \leq 3 + 4(|K|-1) = 4|K| - 1$. We now conclude by Lemma~\ref{lem:ov}.   
\end{proof}

\section{Eccentricities of $k$-subsets of nodes in a tree}\label{sec:kEccTree} 

We now address the problem of computing the eccentricity of {\em subsets} of nodes, in {\em one} tree, with an application to the diameter problem for cube-free median graphs. 

\begin{theorem}\label{thm:main-tree}
Let $T$ be any $n$-node tree, and let $\alpha : V(T) \to \mathbb{R}$.
After a pre-processing in ${\cal O}(n\log{n})$ time, for any subset $U$ of nodes and function $\beta : U \to \mathbb{R}$, we can compute the value $e_{T,\alpha}(U,\beta) = \max_{v \in V(T)}\min_{u \in U}(\alpha(v) + d_T(v,u) + \beta(u))$ in ${\cal O}(|U|\log^2{n})$ time.
\end{theorem}

In particular, after an ${\cal O}(n\log{n})$-time pre-processing we can compute $e_T(U)$ in ${\cal O}(|U|\log^2{n})$ time by setting $\alpha(v) = 0$ for every node $v \in V(T)$ and $\beta(u) = 0$ for every node $u \in U$. 

Let us briefly discuss why we need to consider {\em another} decomposition technique than the centroid decomposition.
Roughly speaking, the centroid decomposition of a tree recursively disconnects the input into balanced subtrees.
After $i$ recursive steps, in the intermediate forest  of subtrees computed until then, it is interesting to keep track, for each subtree, of its contact points with the other subtrees of the forest ({\it i.e.}, the nodes onto the unique path in the input between the subtree considered and the other subtrees of the forest). However, there may be up to $\Omega(i)$ such contact points. In particular, throughout the recursive procedure, the number of contact points per subtree may become non-constant. The classic {\em heavy-path decomposition}~\cite{SlT83} offers us a better control over these contact points. We refer to Fig.~\ref{fig:hp} for a running example. Formally, to every node $v \in V(T)$ that is not a leaf, we associate a child node $u$ such that the number of nodes in the subtree rooted at $u$ is maximized (if they are many possibilities for $u$, then we break ties arbitrarily). Equivalently, we select the edge $uv$. Then, the edges selected during this greedy procedure induce a forest of paths, with each path being called a {\em heavy-path}. The {\em heavy-path tree} (for short, HP-tree) is the rooted tree obtained from $T$ by contracting each heavy path to a single node. Finally, the classic heavy-path decomposition consists in, being given $T$, to output a HP-tree and, for each node of the latter, the corresponding heavy-path in $T$. A heavy-path decomposition can be computed in linear time~\cite{SlT83}. Furthermore, the same as for the centroid decomposition, it can be easily checked that the height of a HP-tree, for an $n$-node rooted tree, is in ${\cal O}(\log{n})$.

\begin{figure}[!ht]
    \centering
    \begin{subfigure}{.45\textwidth}
    	\begin{center}
			\includegraphics[width=\textwidth]{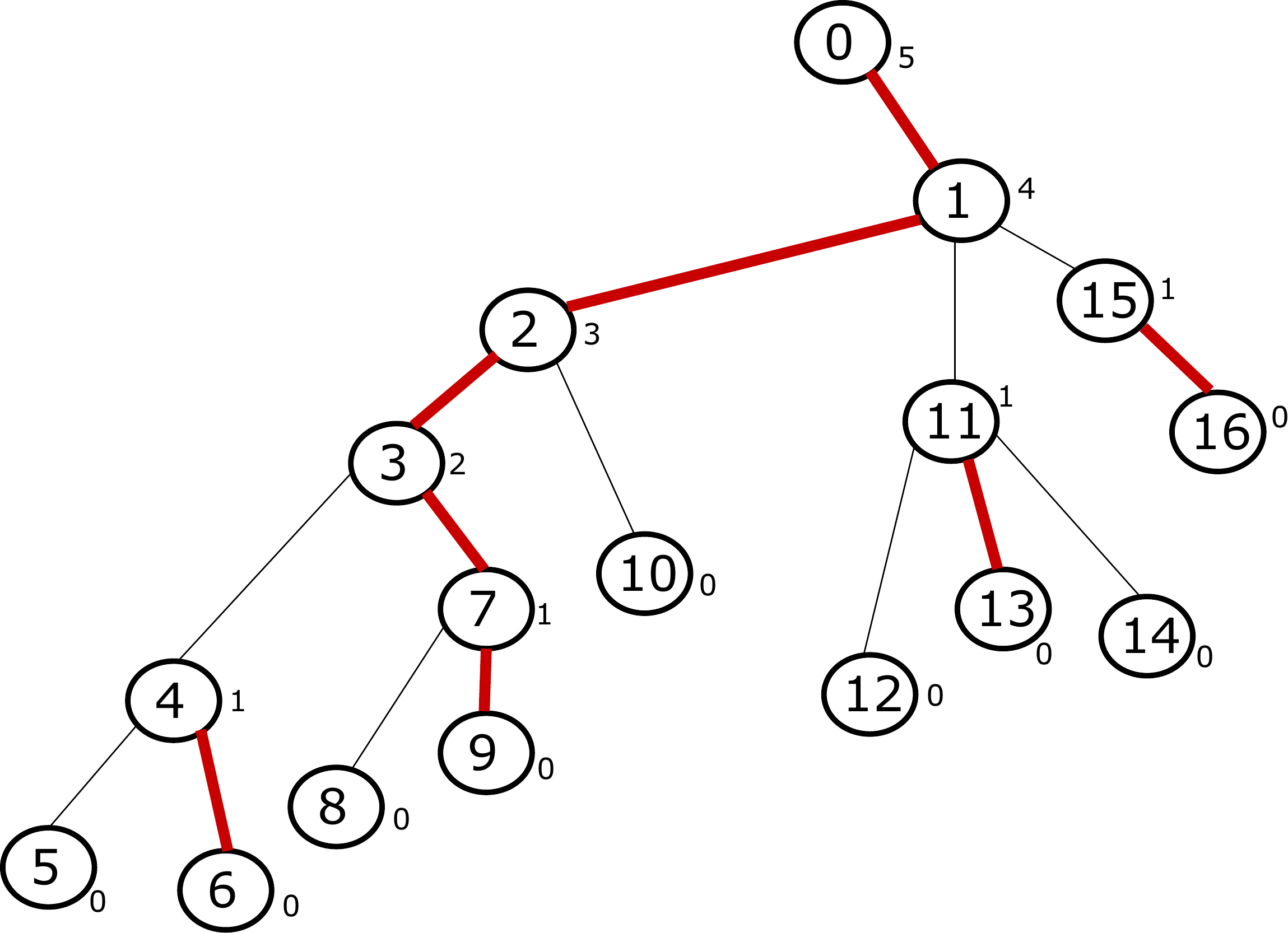}
		\end{center}
        \caption{A tree $T$ (arbitrarily rooted at $0$). Red bold edges are those selected for the heavy-path decomposition. We indicate, next to every node $v$, the height of the corresponding rooted subtree.}
     \end{subfigure}\hfill
     \begin{subfigure}{.4\textwidth}
    	\begin{center}
			\includegraphics[width=\textwidth]{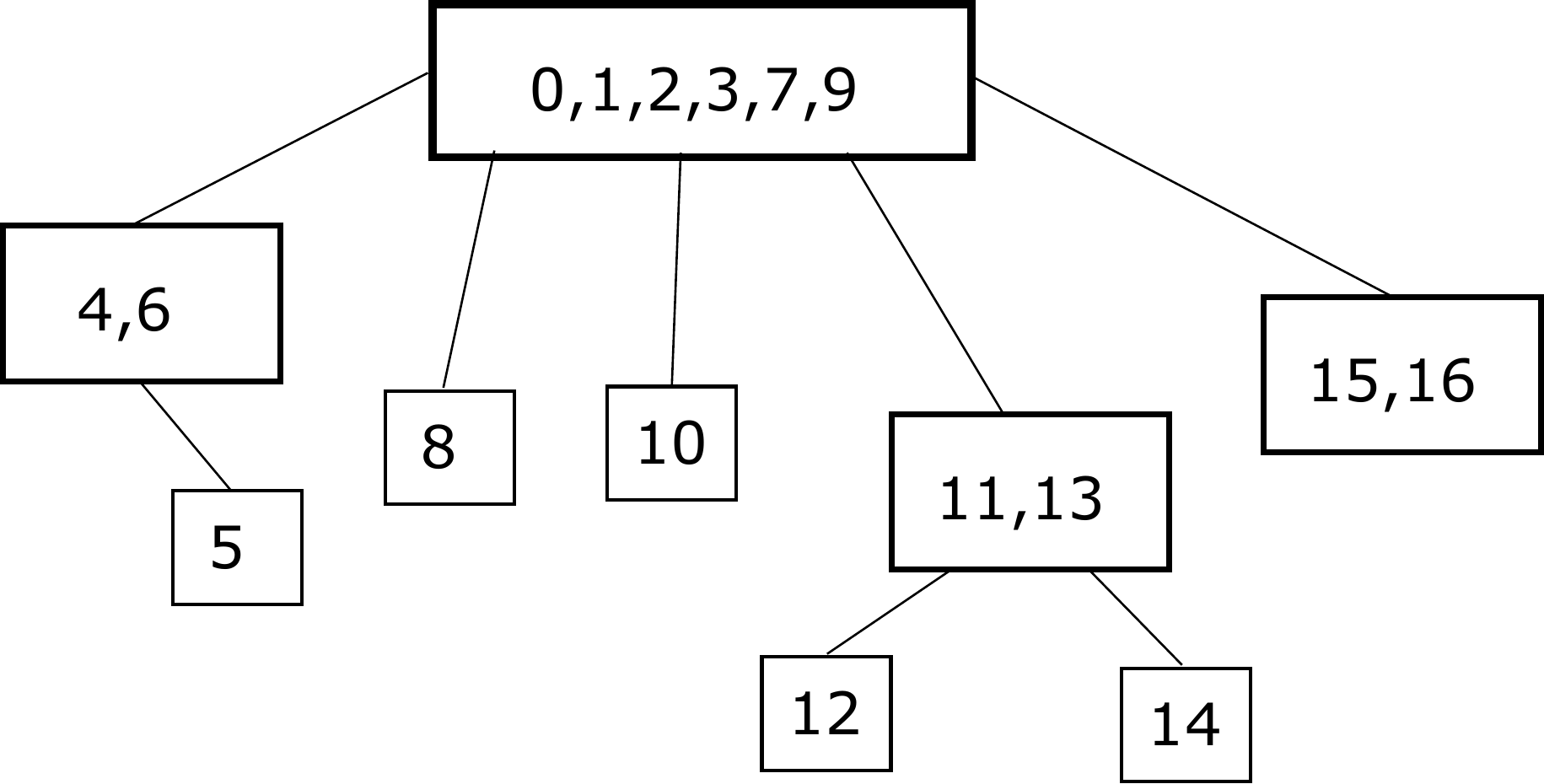}
		\end{center}
        \caption{Heavy-path tree $T'$.}
     \end{subfigure}\vfill
    \caption{Heavy-path decomposition.}
    \label{fig:hp}
\end{figure}

\begin{proof}[Proof of Theorem~\ref{thm:main-tree}.]
During the pre-processing phase, we root the tree $T$ arbitrarily and we compute a {\em heavy-path decomposition}. It takes ${\cal O}(n)$ time. Let us denote by $T'$ the resulting HP-tree and by ${\cal H} \approx V(T')$ the corresponding set of heavy-paths. By using a standard dynamic programming approach, we compute the following set of information for every node $v \in V(T)$ (see also Fig.~\ref{table:preproc}):

\begin{figure}[!ht]
\begin{center}

\begin{subfigure}{.6\textwidth}
\begin{tabular}{l|l|l|l|l|l|l|l}
$v$ & $P_v$ & $\ell(v)$ & $L(v)$ & $h(v)$ & $h^r(v)$ \\
\hline
$0$ & $P_0$  & $0$ &  & $5$ & $0$ \\
\hline
$1$ & $P_0$ & $1$ & $P_5,P_8$ & $4$ & $2$ \\ 
\hline
$2$ & $P_0$ & $2$ & $P_4$ & $3$ & $1$ \\
\hline
$3$ & $P_0$ & $3$ & $P_1$ & $2$ & $2$  \\  
\hline
$4$ & $P_1$ & $0$ & $P_2$ & $1$ & $1$  \\
\hline
$5$ & $P_2$ & $0$ &  & $0$ & $0$  \\ 
\hline
$6$ & $P_1$ & $1$ &  & $0$ & $0$  \\
\hline
$7$ & $P_0$ & $4$ & $P_3$ & $1$ & $1$  \\ 
\hline
$8$ & $P_3$ & $0$ &  & $0$ & $0$ \\
\hline
$9$ & $P_0$ & $5$ &  & $0$ & $0$  \\ 
\hline
$10$ & $P_4$ & $0$ &  & $0$ & $0$  \\
\hline
$11$ & $P_5$ & $0$ & $P_6,P_7$ & $1$ & $1$  \\  
\hline
$12$ & $P_6$ & $0$ &  & $0$ & $0$  \\
\hline
$13$ & $P_5$ & $1$ &  & $0$ & $0$ \\ 
\hline
$14$ & $P_7$ & $0$ &  & $0$ & $0$ \\
\hline
$15$ & $P_8$ & $0$ &  & $1$ & $0$ \\ 
\hline
$16$ & $P_8$ & $1$ &  & $0$ & $0$ 
\end{tabular}
\end{subfigure}\hfill
\begin{subfigure}{.3\textwidth}
\begin{tabular}{l|l|l|l|l}
$P$ & $V(P)$ & $v_P$ & $h(P)$ \\
\hline
$P_0$ & $0,1,2,3,7,9$ & $0$ & $5$  \\
\hline
$P_1$ & $4,6$ & $4$ & $1$ \\
\hline
$P_2$ & $5$ & $5$ & $0$ \\
\hline
$P_3$ & $8$ & $8$ & $0$ \\  
\hline
$P_4$ & $10$ & $10$ & $0$ \\
\hline
$P_5$ & $11,13$ & $11$ & $1$ \\
\hline
$P_6$ & $12$ & $12$ & $0$ \\
\hline
$P_7$ & $14$ & $14$ & $0$ \\  
\hline
$P_8$ & $15,16$ & $15$ & $1$   
\end{tabular}
\end{subfigure}
\caption{Pre-processing.}
\label{table:preproc}
\end{center}
\end{figure}

\begin{enumerate}
\item a pointer to the unique $P_v \in {\cal H}$ s.t. $v \in V(P_v)$. For that, it is sufficient to scan the paths in ${\cal H}$, and to store in an array, for each $v \in V(T)$, a pointer to $P_v$.
\item the root $v_{P_v}$ of $P_v$ (closest node to the root of $T$);
\item the distance $\ell(v) := d_T(v,v_{P_v})$. In particular, the nodes of a heavy-path of length $p$ are numbered $0,1,\ldots,p-1$.
\item the height values $h(v) := \max\{ d_T(v,x) + \alpha(x) \mid x \in V(T_v) \}$ and $h^r(v) = \max\{ d_T(v,x) + \alpha(x) \mid x \in V(T_v^r) \}$, where $T_v$ and $T_v^r$ are the subtrees rooted at $v$ in the tree $T$ and in the forest $T \setminus E(P_v)$, respectively. Note that all the values $h(v), h^r(v)$ can be computed in total ${\cal O}(n)$ time, by using a standard dynamic programming approach. 
\item finally, the list $L(v) := (P^v_1,P^v_2,\ldots,P^v_d)$, ordered by nonincreasing values of $h(P^v_i) := h(v_{P^v_i})$, of pointers to all paths $P^v_i \in {\cal H}$ whose root $v_{P^v_i}$ is a child node of $v$ in $T$. The unordered content of each list $L(v)$ can be computed in total ${\cal O}(n)$ time, as follows: for each $P \in {\cal H}$ we put it in $L(r(P))$, where $r(P)$ denotes the father of its root $v_P$. Then, all lists can be sorted in total ${\cal O}(n\log{n})$ time.
\end{enumerate}
Overall, the above pre-processing phase takes total time in ${\cal O}(n\log{n})$. We end up pre-processing each $P \in {\cal H}$ s.t. the following type of range queries can be answered in ${\cal O}(1)$ time: for each range $(\ell_{\min},\ell_{\max})$, find a $v \in V(P)$ s.t. $\ell(v) \in (\ell_{\min},\ell_{\max})$ and the value $h^r(v) - \ell(v)$ (resp., $h^r(v) + \ell(v)$) is maximized. Specifically, let us consider the sequence $(h^r(v_0),h^r(v_1)-1,\ldots,h^r(v_{p-1})-p+1)$. A {\em Cartesian tree} for that sequence has for root some index $j$ s.t. $h^r(v_j) - j$ is maximized, for left subtree a Cartesian subtree for the sub-sequence $(h^r(v_0),h^r(v_1)-1,\ldots,h^r(v_{j-1})-j+1)$, and for right subtree a Cartesian tree for the sub-sequence $(h^r(v_{j+1})-j-1,h^r(v_{j+2})-j-2,\ldots,h^r(v_{p-1})-p+1)$. See Fig.~\ref{fig:pre-processing} for an illustration. A Cartesian tree can be constructed in ${\cal O}(|V(P)|)$ time, and so, in total ${\cal O}(n)$ time; furthermore, being given a Cartesian tree, we can solve the aforementioned type of range queries in ${\cal O}(1)$ time, by a simple reduction to least-common ancestor queries~\cite{GBT84}. We proceed similarly for the sequence $(h^r(v_0),h^r(v_1)+1,\ldots,h^r(v_{p-1})+p-1)$.

\begin{figure}[!ht]
    \centering
	\includegraphics[width=.4\textwidth]{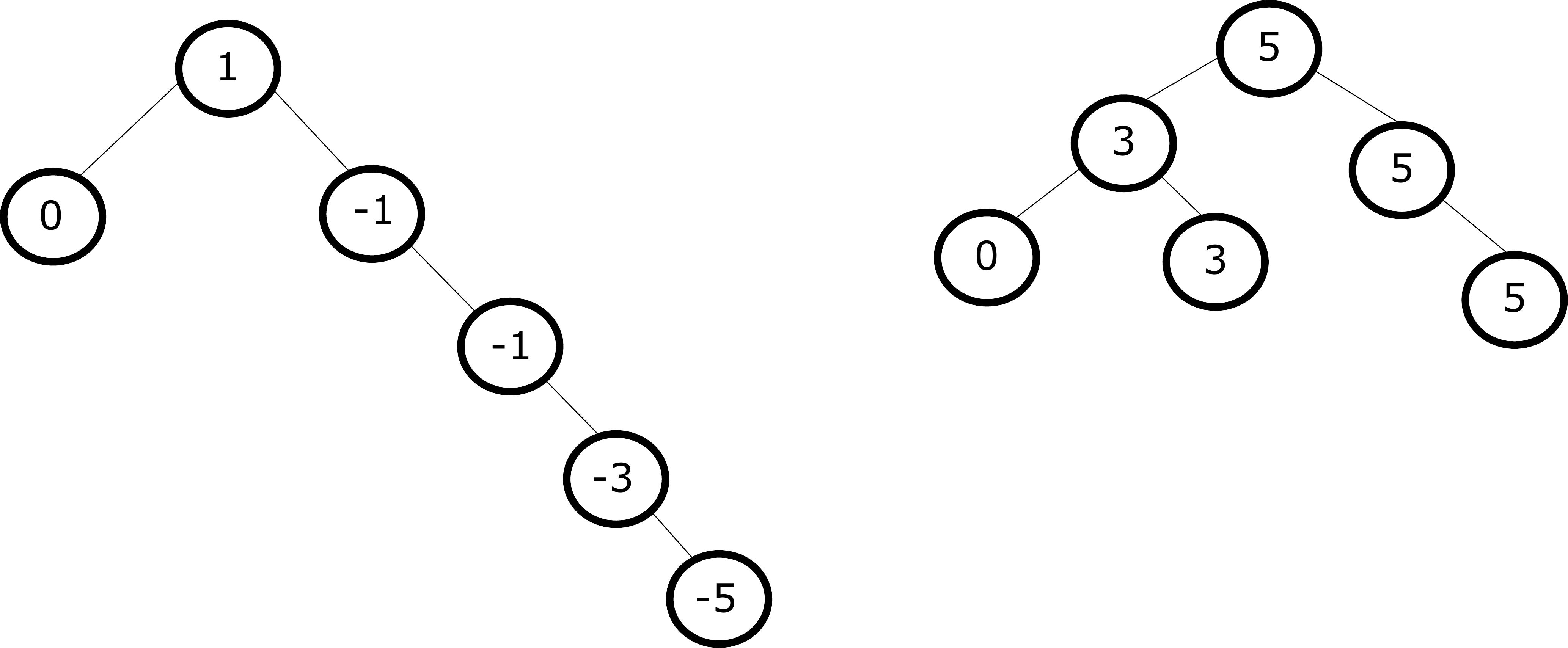}
    \caption{Cartesian trees for the heavy-path $P_0$ (root of the HP-tree).}
    \label{fig:pre-processing}
\end{figure}

\medskip
{\bf Answering a query.}
Let $U \subseteq V(T)$ and $\beta : U \to \mathbb{R}$ be arbitrary. Next, we sketch a recursive procedure over the HP-tree $T'$ in order to compute $e_{T,\alpha}(U,\beta)$.  We observe that some operations below could be optimized, nevertheless we refrained ourselves to do so since the latter operations were not the main bottleneck of the algorithm. In what follows, we denote the root of $T'$ by $P_0$.
We compute $Pr(U,P_0)$, {\it i.e.}, for each node $u \in U$, we add into this set the closest node to $u$ in $P_0$. For that, for all $u \in U$, we construct the $P_uP_0$-path in $T'$, denoted ${\cal Q}_u$. Since the height of $T'$ is in ${\cal O}(\log{n})$, it can be done in ${\cal O}(|U|\log{n})$ time. Let $Pr(u,P_0)$ be the closest node to $u$ in $P_0$. We can compute $Pr(u,P_0)$ from ${\cal Q}_u$ in ${\cal O}(1)$ time. We order $Pr(U,P_0) = (u_1',u_2',\ldots,u_s')$ by increasing $\ell( \cdot )$ ({\it i.e.}, by increasing distance to $v_{P_0}$). Again, it can be done in ${\cal O}(|U|\log{n})$ time. Then, let ${\cal F}_U$ be the forest of $\leq |U|$ subtrees of $T \setminus P_0$ that contain nodes of $U$. We gradually compute $e_{int}(U) := \max_{v \in V(T \setminus {\cal F}_U)} \min_{u \in U} (\alpha(v) + d_T(v,u) + \beta(u))$, then we recurse on each subtree of ${\cal F}_U$. 
More specifically, we decompose the task of computing $e_{int}(U)$ as follows. 

\smallskip
{\it Step 1.} For each $u_i' \in Pr(U,P_0)$, we compute $D(u_i') := \min_{u \in U} (\alpha(u_i') + d_T(u_i',u) + \beta(u))$. For that, there are several intermediate steps.
\begin{enumerate}[label=\alph*]
\item ({\it Going down into the rooted subtree}) We initialize $D(u_i') = \min\{ \beta(u) + d_T(u,u_i') + \alpha(u_i') \mid u_i' = Pr(u,P_0) \}$. If $u_i' \in U$, then $u_i' = Pr(u_i',P_0)$, and we may start with $D(u_i') = \alpha(u_i') + \beta(u_i')$ (otherwise, we start with $D(u_i') = +\infty$). Then, for every $u \in U$ s.t. $Pr(u,P_0) = u_i'$, $u \neq u_i'$, let ${\cal Q}_u = (P_u = P^0, P^1, \ldots, P^d = P_0)$ be the $P_uP_0$-path in $T'$. Recall that, $\forall 0 \leq j \leq d-1$, there is an edge between the root of $P^j$ and its father $r(P^j)$. Then, we have $d_T(u,u_i') = \ell(u) + \left[\sum_{j=1}^{d-1}\ell(r(P^{j-1}))\right] + d$. As a result, we can implement this initialization step in total ${\cal O}(\sum_{u \in U}|{\cal Q}_u|) = {\cal O}(|U|\log{n})$ time.
\item ({\it Going backward onto the heavy-path}) For $i$ from $2$ to $s$, we consider the nodes $u_i' \in Pr(U,P_0)$ sequentially. Note that $d_T(u_i',u_{i-1}') = \ell(u_i') - \ell(u_{i-1}')$. If $D(u_{i-1}') + (\ell(u_i') - \ell(u_{i-1}')) + (\alpha(u_i') - \alpha(u_{i-1}')) < D(u_i')$, then we update $D(u_i') :=  D(u_{i-1}') + (\ell(u_i') - \ell(u_{i-1}')) + (\alpha(u_i') - \alpha(u_{i-1}'))$. It takes ${\cal O}(|Pr(U,P_0)|) = {\cal O}(|U|)$ time.
\item ({\it Going forward onto the heavy-path}) For $i$ from $s-1$ downto $1$, we consider the nodes $u_i' \in Pr(U,P_0)$ sequentially. Note that $d_T(u_i',u_{i+1}') = \ell(u_{i+1}') - \ell(u_i')$. If $D(u_{i+1}') + (\ell(u_{i+1}') - \ell(u_i')) + (\alpha(u_i') - \alpha(u_{i+1}')) < D(u_i')$, then we update $D(u_i') :=  D(u_{i+1}') + (\ell(u_{i+1}') - \ell(u_i')) + (\alpha(u_i') - \alpha(u_{i+1}'))$. It takes ${\cal O}(|Pr(U,P_0)|) = {\cal O}(|U|)$ time.
\end{enumerate}

\smallskip
{\it Step 2.} We search a $v \notin V(P_0 \cup {\cal F}_U)$ s.t. the closest node to $v$ in $P_0$ belongs to $Pr(U,P_0)$ and $\min_{u \in U}(\alpha(v) + d_T(v,u) + \beta(u))$ is maximized. For that, for each $u_i' \in Pr(U,P_0)$, we scan the ordered list $L(u_i')$ until we find a $P_i$ s.t. $P_i \not\subseteq V({\cal F}_U)$. Then, the searched value is exactly $\max_i (h(P_i)+1+D(u_i')-\alpha(u_i'))$. If we can decide in ${\cal O}(1)$ time whether, being given a heavy-path $P_i'$, $P_i' \subseteq V({\cal F}_U)$, then (since in addition there can only be ${\cal O}(|{\cal F}_U|) = {\cal O}(|U|)$ heavy-paths $P_i'$ with this property) this above intermediate step can be done in total ${\cal O}(|U|)$ time. For that, for all $u \in U$, we consider the path ${\cal Q}_u$ and (if $|{\cal Q}_u| \geq 2$) we mark the penultimate path (child of $P_0$ in $T'$). Then, $P_i' \subseteq V({\cal F}_U)$ if and only if $P_i'$ was marked.

\smallskip
{\it Step 3.} We end up searching a $v \notin V({\cal F}_U)$ s.t. its closest node in $P_0$ does {\em not} belong to $Pr(U,P_0)$, and $\min_{u \in U}(\alpha(v) + d_T(v,u) + \beta(u))$ is maximized. For some $1 \leq i < s$, let us consider all such $v$ whose closest node in $P_0$ lies strictly between $u_i',u_{i+1}'$. We denote by $v'$, $\ell(u_i') < \ell(v') < \ell(u_{i+1}')$ the closest node to $v$ in $P_0$. Then, 
\begin{align*}
\min_{u \in U}(\alpha(v) + d_T(v,u) + \beta(u)) &= \alpha(v) + d_T(v,v') \\ 
&+ \min\{d_T(v',u_{i}') + D(u_i') - \alpha(u_i'),d_T(v',u_{i+1}') + D(u_{i+1}') - \alpha(u_{i+1}')\} \\ \\
&\leq h^r(v') + \min\{\ell(v') - \ell(u_i') + D(u_i') - \alpha(u_i'), \ell(u_{i+1}') - \ell(v') + D(u_{i+1}') \\
&- \alpha(u_{i+1}')\}.
\end{align*}
 We set $\ell_{\lim} := \frac{(D(u_{i+1}') - D(u_i')) + \ell(u_i') + \ell(u_{i+1}')+ (\alpha(u_i')- \alpha(u_{i+1}'))}2$. Indeed, $\ell(v') - \ell(u_i') + D(u_i') - \alpha(u_i') \leq \ell(u_{i+1}') - \ell(v) + D(u_{i+1}') - \alpha(u_{i+1}')$ if and only if $\ell(v') \leq \ell_{\lim}$. Therefore, we are left with two range queries: find a $v' \in V(P_0)$, $\ell(u_i') < \ell(v') \leq \ell_{\lim}$, maximizing $h^r(v') + \ell(v')$; resp. find a $v' \in V(P_0)$, $\ell_{\lim} \leq \ell(v') < \ell(u_{i+1}')$, maximizing $h^r(v') - \ell(v')$. Both queries can be answered in ${\cal O}(1)$ time, and so, this last step takes total time in ${\cal O}(|U|)$. The special case of nodes whose closest ancestor in $P_0$ lies before $u_1'$, resp. after $u_s'$, is dealt with in a similar fashion.

\smallskip
After computing $e_{int}(U)$, we consider each subtree of ${\cal F}_U$. Let $T_j$ be such a subtree, it starts with some child $P_j$ of the root path $P_0$ (extracted from the $P_uP_0$-paths, for $u \in U$). Observe that $T_j'$: the subtree of $T'$ rooted at $P_j$, is a HP-tree of $T_j$. In particular, there is no need pre-processing $T_j$. Then, let us define $D_j := \min\{ d_T(v_{P_j},u) + \beta(u) \mid  u \in U \setminus V(T_j) \}$. Finally, we set $U_j := (U \cap V(T_j)) \cup \{v_{P_j}\}$ and $\beta_j : U_j \to \mathbb{R}$ s.t.: 
\begin{itemize}
\item if $u_j \in U_j \setminus \{v_{P_j}\}$, then $\beta_j(u_j) = \beta(u_j)$;
\item if $v_{P_j} \notin U$ then $\beta_j(v_{P_j}) = D_j$, otherwise $\beta_j(v_{P_j}) = \min\{D_j,\beta(v_{P_j})\}$.
\end{itemize}
We apply our query-answering algorithm recursively, for the new input $T_j,U_j,\beta_j$. In particular, we are left explaining the computation of $D_j$. Next, we present a way we can adapt our above techniques in order to compute $e_{int}(U)$ so as to also compute all the values $D_j, \ \text{for} \ T_j \in {\cal F}_U$, in total ${\cal O}(|U|\log{n})$ time:

\smallskip
{\it Step 4.} We consider the nodes $u_i' \in Pr(U,P_0)$, for $1 \leq i \leq s$, sequentially. There are two intermediate distance values to compute, namely:
\begin{align*}D_{\leftarrow}(u_i') &= \min\{ d_T(u_i',u) + \beta(u) \mid u \in U, d_T(u_i',u) = d_T(u_i',u_{i-1}') + d_T(u_{i-1}',u) \}, \\
D_{\rightarrow}(u_i') &= \min\{ d_T(u_i',u) + \beta(u) \mid u \in U, d_T(u_i',u) = d_T(u_i',u_{i+1}') + d_T(u_{i+1}',u) \},
\end{align*}
with the understanding that $D_{\leftarrow}(u_1') = D_{\rightarrow}(u_s') = +\infty$. For that, it suffices to adapt the double-scan routine of Step $1$. Specifically, in order to compute all the values $D_{\leftarrow}(u_i')$, it suffices to execute Steps $1a$ and $1b$ (setting the values $\alpha(u_i')$ to $0$). In the same way, in order to compute all the values $D_{\rightarrow}(u_i')$, it suffices to execute Steps $1a$ and $1c$ (setting the values $\alpha(u_i')$ to $0$). Hence, the total running time of this intermediate step is in ${\cal O}(|U|\log{n})$.

\smallskip
{\it Step 5.} We keep all heavy-paths $P_j \in L(u_i')$ s.t. $P_j \subseteq V({\cal F}_U)$ in a separate sub-list $W(u_i')$. It takes total time in ${\cal O}(|U|)$. We extract from $W(u_i')$ the heavy-path $P_{j_1}$ minimizing $\min\{ d_T(u_i',u) + \beta(u) \mid u \in U \cap V(T_{j_1})  \}$ (if $W(u_i') \neq \emptyset$). We set $D^{(1)}_{\downarrow}(u_i') = \min\{ d_T(u_i',u) + \beta(u) \mid u \in U \cap V(T_{j_1})  \}$. In the same way, we extract from $W(u_i') \setminus \{P_{j_1}\}$ the heavy-path $P_{j_2}$ minimizing $\min\{ d_T(u_i',u) + \beta(u) \mid u \in U \cap V(T_{j_2})  \}$ (if $|W(u_i')| > 1$). We set $D^{(2)}_{\downarrow}(u_i') = \min\{ d_T(u_i',u) + \beta(u) \mid u \in U \cap V(T_{j_2})  \}$. Since we are given the paths ${\cal Q}_u, \ \forall u \in U$, this can be done in total ${\cal O}(|U|\log{n})$ time ({\it i.e.}, see Step $1a$ for the computation of the distances between $u_i'$ and the nodes $u \in U$ s.t. $Pr(u,P_0) = u_i'$).

\smallskip
{\it Step 6.} We set $D_{=}(u_i') = + \infty$ if $u_i' \notin U$, otherwise we set $D_{=}(u_i') = \beta(u_i')$. Then, we consider all heavy-paths $P_j \in W(u_i')$ sequentially. If $P_j \neq P_{j_1}$, then we get: $$D_j = 1 + \min\{D_{=}(u_i'),D_{\leftarrow}(u_i'),D_{\rightarrow}(u_i'),D^{(1)}_{\downarrow}(u_i')\}.$$ Otherwise, we get: $$D_j = 1 + \min\{D_{=}(u_i'),D_{\leftarrow}(u_i'),\\D_{\rightarrow}(u_i'),D^{(2)}_{\downarrow}(u_i')\}.$$

\smallskip
A few more observations are needed in order to prove the query time complexity.
At the $t^{th}$ recursive stage of the algorithm, for $t \geq 0$, we deal with subtrees $T_j$, and corresponding HP-subtrees $T_j'$, with the property that the root of $T_j'$ is at a distance exactly $t$ to the root of $T'$. In particular, the number of recursive stages is no more than the depth of $T'$, and so it is in ${\cal O}(\log{n})$. Furthermore, if we denote by $U_j$ the associated node-subsets, then the running-time of the $t^{th}$ recursive stage is in ${\cal O}((\sum_j|U_j|)\log{n})$. Hence, we are left upper bounding $\sum_j|U_j|$. For that, we observe that, for any fixed $j$, all the nodes of $U_j$, except maybe the root of $T_j$, are also nodes of the original node-subset $U$. Furthermore, since every subtree $T_j$ considered must contain a node of $U$, there are at most $|U|$ of them at any recursive stage. As a result, $\sum_j|U_j| \leq 2|U|$. It implies a complexity in ${\cal O}(|U|\log^2{n})$ for the whole query-answering algorithm.
\end{proof}

\subsection{Application to the diameter problem on cube-free median graphs}\label{sec:diam}

Recall that a graph $G=(V,E)$ is called {\em median} if, for any triple $x,y,z \in V$, there exists a vertex $c$ that is simultaneously on some shortest $xy$-, $yz$- and $zx$-paths. Recently, it was shown that a {\em centroid} in a median graph $G$ can be computed in linear time, with a centroid being any vertex $c$ minimizing $\sum_{v \in V}d_G(v,c)$~\cite{BCCV20}. -- For trees, that are a subclass of median graphs, this definition of centroid is equivalent to the one given in Sec.~\ref{sec:algo}. -- In the technical version of this paper~\cite{BCCVR19}, the authors also asked whether the diameter of a median graph can be computed faster than by reducing to APSP. We make a new step in this direction:

\begin{theorem}\label{thm:cube-free-median}
There is an ${\cal O}(n\log^3{n})$-time algorithm for computing the diameter of $n$-vertex cube-free median graphs.
\end{theorem}

Our proof of this above result heavily relies on a recent {\em distance-labeling scheme} for cube-free median graphs~\cite{CLR20}. We first recall some notions and results from this prior work. In what follows, let $G=(V,E)$ be a cube-free median graph, and let $c \in V$ be a centroid. A subgraph $H$ of $G$ is {\em gated} if, for every $v \in V \setminus V(H)$, there exists a $v^* \in V(H)$ s.t., $\forall u \in V(H), \ d_G(u,v) = d_G(u,v^*) + d_G(v^*,v)$. We define the {\em fibers} $F(x) = \{x\} \cup \{ v \in V(G \setminus H) \mid x \ \text{is the gate of} \ v  \ \text{in} \ H\}$. The fibers $F(x), \ x \in V(H)$ partition the vertex-set of $G$, and each induces a gated subgraph~\cite{CLR20}. For any $z \in V$, the {\em star} $St(z)$ of $z$ is the subgraph of $G$ induced by all edges and squares of $G$ incident to $z$. Any such star $St(z)$ is gated and, if furthermore $z=c$, every fiber $F(x), \ x \in St(c)$ contains at most $|V|/2$ vertices~\cite{CLR20}. A fiber $F(x)$ of the star $St(c)$ is a {\em panel} if $x \in N_G(c)$, and a {\em cone} otherwise. We say that two fibers $F(x),F(y)$ are {\em neighboring} if there exists an edge with an end in $F(x)$ and the other end in $F(y)$. If two fibers are neighboring then one must be a panel and the other must be a cone; furthermore, a cone has two neighboring panels~\cite{CLR20}. Two fibers are {\em $2$-neighboring} if they are cones adjacent to the same panel. Finally, two fibers that are neither neighboring nor $2$-neighboring are called {\em separated}. The subset of vertices in $F(x)$ with a neighbour in $F(y)$ is denoted by $\partial_yF(x)$ (with the understanding that $\partial_yF(x) = \emptyset$ when $F(x),F(y)$ are not neighboring). The {\em total boundary} of $F(x)$ is defined as $\partial^*F(x) = \cup_y \partial_y F(x)$. For a set of vertices $A$, an {\em imprint} of a vertex $u$ is a vertex $a \in A$ such that there is no vertex of $A$ (but $a$ itself) on any shortest $au$-path. A subgraph $H$ of $G$ is {\em quasigated} if every vertex of $V(G \setminus H)$ has at most two imprints. It is known that for each fiber $F(x)$ of a star $St(z)$, the total boundary $\partial^*F(x)$ is an isometric quasigated tree~\cite{CLR20}.

\begin{lemma}[\cite{CLR20}]\label{lem:classification-dist}
Let $G=(V,E)$ be a cube-free median graph, let $c \in V$ be a centroid, and let $F(x),F(y)$ be two fibers of the star $St(c)$. The following hold for every $u \in F(x), \ v \in F(y)$.
\begin{itemize}
\item If $F(x)$ and $F(y)$ are separated, then $d_G(u,v) = d_G(u,c) + d_G(c,v)$;
\item If $F(x)$ and $F(y)$ are neighboring, $F(x)$ is a panel and $F(y)$ is a cone, then let $u_1,u_2$ be the two (possibly equal) imprints of $u$ on the total boundary $\partial^*F(x)$, and let $v^*$ be the gate of $v$ in $F(x)$. We have $d_G(u,v) = \min\{ d_G(u,u_1) + d_G(u_1,v^*) + d_G(v^*,v), d_G(u,u_2) + d_G(u_2,v^*) + d_G(v^*,v) \}$;
\item If $F(x)$ and $F(y)$ are $2$-neighboring, then let $F(w)$ be the panel neighboring $F(x)$ and $F(y)$. Let $u^*$ and $v^*$ be the gates of $u$ and $v$ in $F(w)$. Then $d_G(u,v) = d_G(u,u^*) + d_G(u^*,v^*) + d_G(v^*,v)$.
\end{itemize}
\end{lemma}

In what follows, we sometimes use the fact that cube-free median graphs are sparse~\cite{CLR20}.

\begin{proof}[Proof of Theorem~\ref{thm:cube-free-median}]
We first compute a centroid $c$ for $G$, that takes ${\cal O}(n)$ time~\cite{BCCVR19}. We compute a breadth-first search rooted at $c$, thereby computing $d_G(c,v)$ for every $v \in V$. Note that the eccentricity value $e_G(c)$ is a lower bound for $diam(G)$. Then, we use an algorithmic procedure from~\cite{CLR20} in order to compute in ${\cal O}(n)$ time the fibers $F(x), \ x \in St(c)$.

{\bf Step 1.} We compute the maximum distance between vertices that are in different panels. For that, we enumerate all the panels $F(x), \ x \in N_G(c)$. We set $D(x) = \max_{v \in F(x)} d_G(v,c)$. This can be done in ${\cal O}(|F(x)|)$ time because the fibers were pre-computed and we performed a BFS rooted at $c$. Recall that two panels are always separated. Then, let $x_1 \in N_G(c)$ be s.t. $D(x_1)$ is maximized, and let $x_2 \in N_G(c) \setminus \{x_1\}$ be s.t. $D(x_2)$ is maximized. By Lemma~\ref{lem:classification-dist}, the maximum distance between vertices that are in different panels is equal to $D(x_1) + D(x_2)$. The total running time for computing this distance is in ${\cal O}(n)$.

{\bf Step 2.} We compute the maximum distance between vertices that are in separated fibers. By the previous step, we may restrict ourselves to the case when at least one vertex is in a cone. For that, 
we assign  to all the fibers $F(x)$ pairwise different positive numbers, that we abusively identify with the vertices $x$ of the star $St(c)$. We enumerate the cones $F(y), \ y \in St(c) \setminus N_G[c]$. Let $F(x),F(x')$ be the two panels neighboring $F(y)$. We create a point $\overrightarrow{p}(y) = (x,x')$, to which we associate the value $f(\overrightarrow{p}(y)) = D(y)$. We put these points, and their associated values, in a $2$-dimensional range tree, that takes ${\cal O}(n\log{n})$ time.
\begin{itemize}
\item We enumerate all the panels $F(x), \ x \in N_G(c)$. In order to compute the maximum distance between a vertex of $F(x)$ and a vertex in any separated cone, it suffices to compute a point $\overrightarrow{p}(y) = (p_1,p_2)$ s.t.: $p_1 \neq x, \ p_2 \neq x \ \text{and} \ f(\overrightarrow{p}(y)) \ \text{is maximized for these properties.}$
Indeed, this maximum distance is exactly $D(x) + D(y) = D(x) + f(\overrightarrow{p}(y))$.
Furthermore, since each inequality gives rise to two disjoint intervals to which the corresponding coordinate must belong, a point $\overrightarrow{p}(y)$ as above can be computed using four range queries. It takes ${\cal O}(\log{n})$ time per query, and so ${\cal O}(\log{n})$ time in total.  
\item Again, we enumerate all the cones $F(y), \ y \in St(c) \setminus N_G[c]$. Let $F(x),F(x')$ be the two panels neighboring $F(y)$. In order to compute the maximum distance between a vertex of $F(y)$ and a vertex in any separated cone, it suffices to compute $\overrightarrow{p}(y') = (p_1,p_2)$ s.t.:
$p_1 \notin \{x,x'\}, \ p_2 \notin \{x,x'\} \ \text{and} \ f(\overrightarrow{p}(y')) \ \text{is maximized for these above properties.}$
Indeed, these are necessary and sufficient conditions for the two cones $F(y),F(y')$ not to be $2$-neighboring. In particular, this maximum distance is exactly $D(y) + D(y') = D(y) + f(\overrightarrow{p}(y'))$. Having two forbidden values for a coordinate gives rise to three disjoint intervals to which the latter may belong. Therefore, a point $\overrightarrow{p}(y')$ as above can be computed using nine range queries. 
\end{itemize}
As a result, the total running time of this step is in ${\cal O}(n\log{n})$.

{\bf Step 3.} We compute the maximum distance between vertices that are in neighboring fibers. For that, we enumerate the panels $F(x), \ x \in N_G(c)$. Let $T = \partial^*F(x)$. Note that all total boundaries can be computed in total ${\cal O}(n)$ time~\cite{CLR20}. For each node $z \in V(T)$, let $\alpha(z) = \max\{ d_G(v,z) \mid v \in F(y), \ F(y) \ \text{and} \ F(x) \ \text{are neighboring}, \ z \ \text{is the gate of} \ v \}$ (with the understanding that, if no vertex has $z$ as its gate, then $\alpha(z)$ is a sufficiently large {\em negative} value, say $\alpha(z) = -|V(T)|$). There is an algorithmic procedure in order to compute, for all vertices $v$ in a cone, their gates in the neighboring panels {\em and} their respective distances to the latter, in total ${\cal O}(n)$ time~\cite{CLR20}. Therefore, after an initialization phase in ${\cal O}(|V(T)|) = {\cal O}(|F(x)|)$ time, all the nodes weights can be computed in ${\cal O}(\sum_{y \mid x \sim y}|F(y)|)$ time, where $x \sim y$ denotes that $F(x),F(y)$ are neighboring. Note that a cone is neighboring two panels~\cite{CLR20}. As a result, we scan each cone twice, and the total running-time of this intermediate phase is in ${\cal O}(n)$.
Then, we enumerate the nodes $u \in F(x)$. We want to compute the value $\max\{ d_G(u,v) \mid v \in F(y), \ y \sim x \}$. For that, let $u_1,u_2 \in V(T)$ be the two (possibly equal) imprints of $u$. Set $U = \{u_1,u_2\}$ and $\beta(u_1) = d_G(u,u_1), \ \beta(u_2) = d_G(u,u_2)$. -- Again, note that for all vertices $u$ in a panel, the corresponding $U$ and $\beta$ can be computed in total ${\cal O}(n)$ time~\cite{CLR20}. -- Since $T$ is isometric, then it follows from the distance formulae in Lemma~\ref{lem:classification-dist} that the value to be computed is exactly $e_{T,\alpha}(U,\beta)$. By Theorem~\ref{thm:main-tree}, the latter can be computed in ${\cal O}(\log^2{n})$ time, up to some initial pre-processing of the tree $T$ in ${\cal O}(|V(T)|\log{|V(T)|}) = {\cal O}(|F(x)|\log{n})$ time. Overall, the running-time of this step is in ${\cal O}(n\log^2{n})$. 

{\bf Step 4.} We compute the maximum distance between vertices that are in $2$-neighboring cones. For that, we consider the panels $F(x), \ x \in N_G(c)$ sequentially, computing the maximum distance between the vertices that in different cones neighboring $F(x)$. Again, let $T = \partial^*F(x)$. The following values are pre-computed for every node $z \in V(T)$:
\begin{itemize}
\item $\alpha(z)$: the maximum distance between $z$ and any vertex $v$ in a neighboring cone $F(y)$ of which $z$ is the gate in $F(x)$ (possibly $\alpha(z) = -|V(T)|$, see also Step 3).
\item $r(z)$: the index $y$ of some neighboring cone $F(y)$ s.t. there is a $v \in F(y)$ of which $z$ is the gate in $F(x)$ and the distance $d_G(v,z) = \alpha(z)$ is maximized.
\item $\alpha'(z)$: the maximum distance between $z$ and any vertex $v$ in a neighboring cone $F(y)$, $y \neq r(z)$, of which $z$ is the gate in $F(x)$ (possibly, $\alpha'(z) = -|V(T)|$).
\end{itemize}
As for Step 3, the whole pre-computation phase runs in ${\cal O}(n)$ time.
Our intermediate objective is to compute the following values for every node $z \in V(T)$:
\begin{itemize}
\item $e_1(z) := \max\{d_T(z,z') + \alpha(z') \mid z' \in V(T)\}$;
\item $s(z) := r(z')$ for some $z' \in V(T)$ s.t. $d_T(z,z') + \alpha(z') = e_1(z)$ is maximized;
\item $e_2(z) := \max\{d_T(z,z') + \alpha(z') \mid z' \in V(T), r(z') \neq s(z) \} \cup \{d_T(z,z') + \alpha'(z') \mid z' \in V(T), r(z') = s(z) \}$. 
\end{itemize}
Being given these above three values, we enumerate all neighboring cones $F(y)$. For every $v \in F(y)$, let $z \in V(T)$ be its gate. If $s(z) \neq y$, then by Lemma~\ref{lem:classification-dist}, the maximum distance between $v$ and a vertex in another cone neighboring $F(x)$ is equal to $d_G(v,z) + e_1(z)$. Otherwise, it is  $d_G(v,z) + e_2(z)$. Therefore, in additional ${\cal O}(\sum_{y \mid x \sim y}|F(y)|)$ time (and so, in total ${\cal O}(n)$ time over all the panels $F(x)$), we can compute the maximum distance between the vertices that in different cones neighboring $F(x)$. We are left computing $(e_1(z),s(z),e_2(z))$ for every node $z \in V(T)$. As it turns out, we can do so in ${\cal O}(n)$ time by slightly adapting a standard dynamic programming approach in order to compute all eccentricities in a tree. Specifically, let us root $T$ arbitrarily.
\begin{itemize}
\item We proceed $T$ bottom-up in order to compute the values $(e_1^{\downarrow}(z),s^{\downarrow}(z),e_2^{\downarrow}(z))$, whose definitions are the same as for $(e_1(z),s(z),e_2(z))$ but restricted to the subtree rooted at $z$. If $z$ is a leaf, then $e_1^{\downarrow}(z) = \alpha(z), \ s^{\downarrow}(z) = r(z) \ \text{and} \ e_2^{\downarrow}(z) = \alpha'(z)$. Otherwise, let $z_1,z_2,\ldots,z_d$ be the children of $z$ in $T$. In ${\cal O}(d)$ time, we select a $i$ s.t. $e_1^{\downarrow}(z_i)$ is maximized; we set $e_1^{\downarrow}(z) = \max\{\alpha(z),1+e_1^{\downarrow}(z_i)\}$, with the corresponding value $s^{\downarrow}(z)$ being set accordingly to either $r(z)$ or $s^{\downarrow}(z_i)$. If $r(z) = s^{\downarrow}(z)$ then we replace $\alpha(z)$ by $\alpha'(z)$. In the same way, for any $1 \leq i \leq d$ s.t. $s^{\downarrow}(z_i) = s^{\downarrow}(z)$, we replace $e_1^{\downarrow}(z_i)$ by $e_2^{\downarrow}(z_i)$. In doing so, we can proceed as before in order to compute $e_2^{\downarrow}(z)$. Overall, the whole intermediate phase runs in ${\cal O}(|V(T)|) = {\cal O}(|F(x)|)$ time.

\item We proceed $T$ up-bottom in order to compute the values $(e_1^{\uparrow}(z),s^{\uparrow}(z),e_2^{\uparrow}(z))$, whose definitions are the same as for $(e_1(z),s(z),e_2(z))$ but restricted to the subtree $T'$ where we removed all descendants of $z$ (including $z$ itself). If $z$ is the root, then (since the corresponding $T'$ is empty) $e_1^{\uparrow}(z) = e_2^{\uparrow}(z) = -\infty$ and we may set $s^{\uparrow}(z)$ arbitrarily. Then, let $z$ be already processed, and let $z_1,z_2,\ldots,z_d$ be the children of $z$. W.l.o.g., $e_1^{\downarrow}(z_1) = \max_{1 \leq i \leq d} e_1^{\downarrow}(z_i)$, and similarly $e_1^{\downarrow}(z_2) = \max_{2 \leq i \leq d} e_1^{\downarrow}(z_i)$. The latter can be ensured in ${\cal O}(d)$ time by re-ordering the nodes. Furthermore, $e_1^{\uparrow}(z_1) = \max\{1+\alpha(z),1+e_1^{\uparrow}(z),2+e_1^{\downarrow}(z_2)\}$, while for every $2 \leq i \leq d$ we have $e_1^{\uparrow}(z_i) = \max\{1+\alpha(z),1+e_1^{\uparrow}(z),2+e_1^{\downarrow}(z_1)\}$. The corresponding values $s^{\uparrow}(z_i), \ 1 \leq i \leq d$ are set accordingly. 

There are only four possibilities for the values $s^{\uparrow}(z_i), \ 1 \leq i \leq d$, namely: $r(z), \ s^{\uparrow}(z) \ \text{and} \\ \text{either} \ s^{\downarrow}(z_1) \ \text{or} \ s^{\downarrow}(z_2)$. Let $y$ be any of these four possibilities. We update the node-weights as before in order to compute the values $e_2^{\uparrow}(z_j)$, for any child $z_j$ s.t. $s^{\uparrow}(z_j) = y$; {\it i.e.}, if $r(z) = y$ then we replace $\alpha(z)$ by $\alpha'(z)$, if $s^{\uparrow}(z) = y$ then we replace $e_1^{\uparrow}(z)$ by $e_2^{\uparrow}(z)$, and for any $1 \leq i \leq d$ s.t. $s^{\downarrow}(z_i) = y$ we replace $e_1^{\downarrow}(z_i)$ by $e_2^{\downarrow}(z_i)$. Overall, the processing of all the children of node $z$ requires ${\cal O}(d)$ time, and therefore the total running time for this up-bottom processing is in ${\cal O}(|V(T)|) = {\cal O}(|F(x)|)$.
\end{itemize}

Finally, for any $z \in V(T)$ we set $e_1(z) = \max\{e_1^{\downarrow}(z),e_1^{\uparrow}(z)\}$, with the corresponding value $s(z)$ being set as one of $s^{\downarrow}(z),s^{\uparrow}(z)$ accordingly. If $s^{\uparrow}(z) = s^{\downarrow}(z)$, then $e_2(z) = \max\{e_2^{\downarrow}(z),e_2^{\uparrow}(z)\}$; otherwise, either $e_1(z) = e_1^{\downarrow}(z)$ and so $e_2(z) = \max\{e_1^{\uparrow}(z),e_2^{\downarrow}(z)\}$, or $e_1(z) = e_1^{\uparrow}(z)$ and so $e_2(z) = \max\{e_2^{\uparrow}(z),e_1^{\downarrow}(z)\}$. Overall, this above dynamic program runs in ${\cal O}(|V(T)|) = {\cal O}(|F(x)|)$ time.

{\bf Step 5.} We compute the maximum distance between two vertices that are in a same fiber. Since the fibers $F(x), \ x \in St(c)$ are gated, and so convex, they all induce a cube-free median graph. In particular, we may apply our algorithm recursively on each fiber. Each fiber contains at most $n/2$ vertices, and so there are ${\cal O}(\log{n})$ recursive stages. Since a stage runs in ${\cal O}(n\log^2{n})$ time (the bottleneck being Step 3), the total running time for computing the diameter is in ${\cal O}(n\log^3{n})$. 
\end{proof}

Interestingly, there also exists a distance-labeling scheme with polylogarithmic labels for {\em bounded-degree} median graphs~\cite{CLR20}. We briefly explain how to derive from the latter an ${\cal O}(2^{{\cal O}(2^{{\cal O}(\Delta)})}n\log{n})$-time algorithm in order to compute the diameter of $n$-vertex median graphs of maximum degree $\Delta$. As for Theorem~\ref{thm:main-tree}, let us consider a more general version of the problem with a graph $G=(V,E)$ and a vertex-weight function $\alpha : V \to \mathbb{R}$. We want to compute $\max_u\max_v(d_G(u,v) + \alpha(v))$. For that, we start computing a centroid $c$ w.r.t. the hop distance ({\it i.e.}, ignoring the vertex-weights). It takes ${\cal O}(\Delta  n)$ time~\cite{BCCV20}. Now, the {\em star} $St(z)$ can be defined, for general median graphs, as the subgraph induced by all {\em hypercubes} containing $z$. If the maximum degree is $\Delta$ then, the maximum dimension of a hypercube containing $z$ is at most $\Delta$. As a result, we can compute $St(z)$ in ${\cal O}(\Delta n+2^{{\cal O}(2^{{\cal O}(\Delta)})})$ time, as follows. We compute in ${\cal O}(\Delta n)$ time all the vertices at distance $\leq \Delta$ to $z$ (there are $\leq \Delta(\Delta-1)^{\Delta-1} = 2^{{\cal O}(\Delta\log{\Delta})}$ such vertices). Then, we enumerate all possible subsets of at most $2^\Delta$ vertices including $z$ and we check whether the latter induce hypercubes. It also follows from~\cite{CLR20} that $St(z)$ is gated, and that if $z=c$, then each fiber $F(x), \ x \in St(c)$ contains at most $n/2$ vertices. Recall that fibers are gated subgraphs~\cite{CLR20}. For every fiber $F(x)$, we compute the gate $v^*$ of every $v \notin F(x)$ and the distance $d_G(v,v^*)$. This takes ${\cal O}(\Delta n)$ by using a modified breadth-first search, and so ${\cal O}(|St(c)|\Delta n)$ time in total. We update the weights of the vertices $u \in F(x)$, as follows: $\alpha_x(u) := \max\{\alpha(u)\} \cup \{d_G(u,v) + \alpha(v) \mid v \notin F(x), \ u \ \text{is the gate of} \ v\}$. Then, applying our algorithm recursively to $F(x)$, we compute $\max_{u \in F(x)}\max_{v \in V}(d_G(u,v) + \alpha(v))$.

\medskip
We would like to summarize our results for the diameter problem on median graphs. The diameter of a given $n$-vertex median graph can be computed in ${\cal O}(n^{1+o(1)})$ time, unless: it has dimension at least three, maximum degree $\Omega(\log\log{n})$, and either lattice dimension $\Omega(\log{n^{1-o(1)}})$ or isometric dimension $\Omega(n^{4/5 - o(1)})$.

\bibliography{tree}

\end{document}